\documentclass[a4paper]{article}
\usepackage{color}
\usepackage{graphicx}
\graphicspath{{.}{figures/}{images/}}
\usepackage[english]{babel}
\usepackage[utf8]{inputenc}
\usepackage{latexsym}
\usepackage{amsmath}
\usepackage{amsfonts}
\usepackage{wasysym}
\usepackage{lscape}
\usepackage{url}

\usepackage[amsmath,thmmarks,thref,hyperref]{ntheorem}
\numberwithin{equation}{section}

\theoremstyle{plain}
\theoremnumbering{arabic}

\theoremheaderfont{\normalfont\bfseries}
\newtheorem{definition}{Definition}
\newtheorem{lemma}[definition]{Lemma}

\newtheorem{theorem}[definition]{Theorem}

\theorembodyfont{\upshape}

\newtheorem{remark}[definition]{Remark}

\theoremstyle{break}
\theorembodyfont{\upshape}
\theoremseparator{}

\theoremstyle{nonumberplain}
\theoremheaderfont{\normalfont\itshape}
\theorembodyfont{\normalfont}
\theoremseparator{.}
\theoremsymbol{\ensuremath{\square}}
\newtheorem{proof}{Proof}
\numberwithin{definition}{section}

\newcommand{\N}[1][]{\mathbb{N}^{#1}}
\newcommand{\Z}[1][]{\mathbb{Z}^{#1}}

\newcommand{\R}[1][]{\mathbb{R}^{#1}}

\newcommand{\vv}{\vec{v}}

\newcommand{\mattwo}[4]{\begin{bmatrix}{#1}&{#2}\\{#3}&{#4}\end{bmatrix}}
\newcommand{\vectwo}[2]{\left(\begin{array}{c}#1\\#2\end{array}\right)}

\newcommand{\ra}{\ensuremath{\rightarrow}}

\newcommand{\Ra}{\ensuremath{~\Rightarrow~}}
\newcommand{\Lra}{\ensuremath{~\Leftrightarrow~}}

\begin{document}

\title{Finding and Classifying Critical Points of 2D Vector
Fields: A Cell-Oriented Approach Using Group Theory}

\author{Felix~Effenberger \and
        Daniel~Weiskopf 
}

\date{April 2010}

\maketitle

\begin{abstract}
We present a novel approach to finding critical points in cell-wise
barycentrically or bilinearly interpolated vector fields on
surfaces.
The Poincar\'{e} index of the critical points is determined by investigating
the qualitative behavior of 0-level sets of the interpolants of the
vector field components in parameter space using precomputed combinatorial results, thus avoiding the computation of the Jacobian of the vector field at the critical points in order to determine its index. The locations of the critical points within a cell
are determined analytically to achieve accurate results. This
approach leads to a correct treatment of cases with two first-order
critical points or one second-order critical point of bilinearly
interpolated vector fields within one cell, which would be missed by
examining the linearized field only. We show that for the considered interpolation schemes determining the
index of a critical point can be seen as a coloring problem of cell
edges. A complete classification of all possible colorings
in terms of the types and number of critical points yielded by each
coloring is given using computational group theory. We
present an efficient algorithm that makes use of these precomputed
classifications in order to find and classify critical points in a
cell-by-cell fashion. Issues of numerical stability, construction of
the topological skeleton, topological simplification, and the
statistics of the different types of critical points are also
discussed.

\end{abstract}

\section{Introduction}
\label{sec:intro}

The visualization of vector field topology is a problem
that arises naturally when studying the qualitative structure of
flows that are tangential to some surface. As usual, we use the term surface for a real, smooth 2-manifold (equipped with an atlas consisting of charts), see for example \cite{oneill83srg} for an introduction to Riemannian Geometry. Having its roots in the theory of dynamical systems, the topological skeleton of a Hamiltonian flow on a surface with isolated critical
points consists of these critical points and trajectories (streamlines) of the
vector field that lie at the boundary of a hyperbolic sector and connect two of the critical points. Helman and Hesselink~\cite{helman90srf} introduced the concept of the topology of a planar vector field to the visualization community and proposed the
following construction scheme: (1) critical points are located, (2)
classified, and then (3) trajectories along hyperbolic sectors are
traced and connected to their originating and terminating critical
points or boundary points. Step (2)---the classification of a
critical point---is usually based on the Jacobian of the vector
field, see~\cite{book-perko91}. Trajectories of step (3) are typically
constructed by solving an ordinary differential equation for
particle tracing.

Although a vast body of previous work in the field of flow
visualization focuses on the problem of how to extend the method of Helman and Hesselink to vector fields on arbitrary surfaces as well as the second and third step of the above algorithm, not much attention has been paid to the first step. In this paper, we specifically address the identification and
classification of critical points in parameter space.

As efficient computer-based visualization algorithms usually work with discrete parametrized versions of the surfaces involved --- examples of popular discretization schemes are for example triangulated or quadrangulated versions of the surface --- we will in this paper not focus on the well-researched field of how to parametrize a given surface (see \cite{floater05sps} for a recent survey) but assume that a surface always comes equipped with a globally continuous discrete parametrization that allows a cell-wise (local) barycentric or bilinear interpolation scheme of a vector field tangential to the surface in parameter space. 

While this task is rather easy
for linear vector fields, the problem setting becomes more
interesting for bilinearly interpolated fields. Bilinear
interpolation is ubiquitous in scientific visualization because it
is popular for widely used uniform or curvilinear grids
representing planar or curved surfaces. Since bilinear
interpolation is not linear, it can lead to higher-order critical
points, which are neglected by often-used linearization
approaches.

\begin{figure*}[t]
\centering
\begin{tabular}{@{}ccc@{}}
  \includegraphics[width=0.2\linewidth]{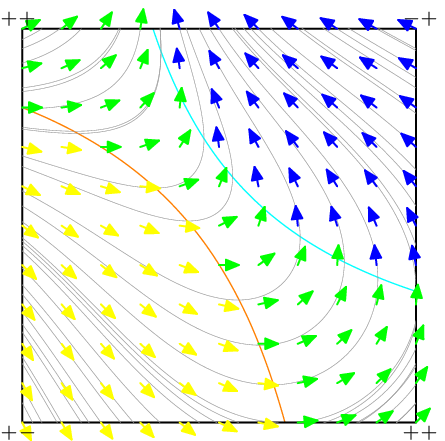} &
  \includegraphics[width=0.2\linewidth]{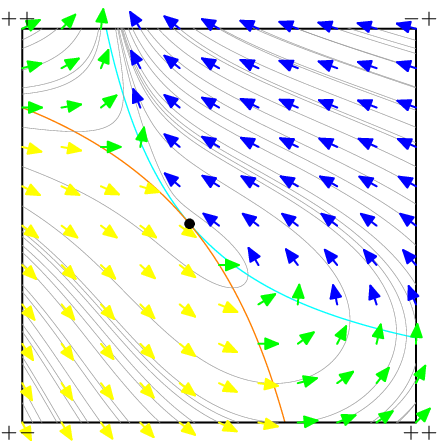} &
  \includegraphics[width=0.2\linewidth]{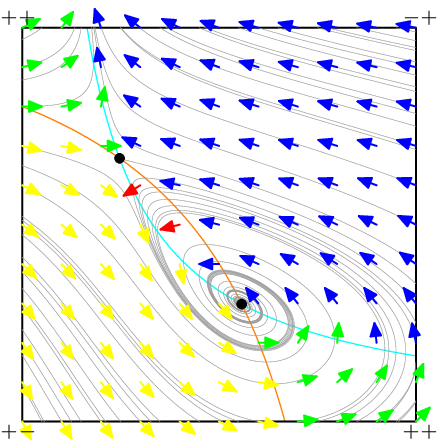}\\
  (a)&(b)&(c)\\
\end{tabular}
  \caption{
Classification of critical points for three intersection cases for
the $0$-level sets of the two vector field components $f_1$ (cyan)
and $f_2$ (orange) of a bilinearly interpolated vector field
$\vec{f}=(f_1,f_2)^T$: (a) no intersection of the level sets, (b)
touching of the level sets yielding one critical point, (c) double
intersection of the level sets yielding two critical points (one
saddle and one non-saddle). The colors of the vector arrows encode
the types of characteristic areas as defined in
Section~\ref{sec:characteristicareas} (green=$++$, yellow=$+-$,
blue=$-+$, red=$--$), a set of streamlines is shown in gray, and
critical points lying in the intersection set of the two level sets
are shown as black dots. Each vertex of a square is marked with
$++$, $+-$, $-+$, $--$ according to the sign of $f_1$ and $f_2$ at
that vertex. }
 \label{fig:intersectioncasesaabb}
\end{figure*}

In this paper, we introduce a new method that locates and classifies
all critical points within piecewise linearly (barycentrically) or
bilinearly interpolated two-dimensional grids.
Our method determines the index of
a critical point without the need to evaluate the Jacobian of the
vector field in the critical point in order to determine its Poincar\'{e} index. For the case of bilinearly interpolated vector fields,
our method is able to detect higher-order critical points and the
presence of two first-order critical points within one cell, which,
to our knowledge, has not been achieved with the common methods \cite{helman91vvft} for the bilinear interpolation scheme before.
Figure~\ref{fig:intersectioncasesaabb} shows a corresponding example
and illustrates our classification method. Our approach is based on
the idea that investigating the qualitative behavior of 0-level sets
of the components of the interpolated vector field provides
information needed to compute the Poincar\'{e} index of a critical
point. All qualitatively different possibilities of this behavior
and the types of critical points yielded by each possibility are
completely classified using the computational group
theory tool GAP (\textbf{G}roups, \textbf{A}lgorithms, and
\textbf{P}rogramming)~\cite{GAP4}. See the enumeration of cases for marching cubes and generic substitope algorithms by Banks et al.~\cite{banks04ccs}
for a previous example of an application of computational group theory in the field of scientific visualization. 
Furthermore, we discuss a cell-based topology simplification method as well as the question of numerical stability. Our approach results in an efficient, accurate, and robust cell-based algorithm for detecting all critical points of
barycentrically or bilinearly interpolated 2D vector fields.

The paper is organized as follows. First, we will give a short
review of the visualization literature dealing with vector field
topology. Then, the theoretical foundations of vector field
topology, namely the theory of the qualitative behavior of
second-order dynamical systems along with such fundamental notions as
those of critical points, separatrices, and the Poincar\'{e} index
of a critical point are reviewed. Following this, we will present
our new approach---first the general framework will be discussed and
then applied to two cases: barycentrically and bilinearly
interpolated vector fields. Then, we will deal with open issues such
as critical points on the boundary of cells and numeric stability
followed by a more detailed description of the cell-based algorithm.
We will conclude giving results and a short review of our method.

This paper has accompanying material in the form of online resources, namely the GAP programs used in this paper (Online Resource 3) and lists of equivalence classes of colored cells referred to in Theorem~\ref{th:classificationbarycentric} (Online Resource 1) and Theorem~\ref{th:classificationbilinear} (Online Resource 2).

%-------------------------------------------------------------------
\section{Previous Work}

Topology-based methods for planar vector fields were first proposed to
the visualization community by Helman and
Hesselink~\cite{helman90srf}, employing methods from the theory of
dynamical systems~\cite{book-andronov73} to planar linear (or
linearized) vector fields in order to visualize flow
characteristics. Their work has triggered a large body of further
research on topology-based flow visualization, an overview of
which is given by Post et al.~\cite{post03saf} and Scheuermann and
Tricoche~\cite{scheuermann05tmf}.

The case of a planar linear vector field is relatively simple to
deal with, but it has a couple of drawbacks, most notably that only
first-order critical points can be detected. Much effort has been
put into methods to overcome this drawback imposed by the
interpolation scheme and to addresses the
issue of detecting and processing higher-order critical points of
interpolated vector fields, also of vector fields on arbitrary surfaces \cite{laramee09btbfva}. 
Such higher-order critical points can be
found, for example, by using piecewise linear interpolation schemes
in combination with a clustering of first-order critical points
according to Tricoche et al.~\cite{tricoche00tsm}. Another example
is the work by Theisel \cite{theisel02dvf}, who proposes a method
for designing piecewise linear planar vector fields of arbitrary
topology. Nonlinear interpolation schemes have been investigated by
Scheuermann et al.~\cite{scheuermann97vnv} and Zhang et
al.~\cite{zhang04vfd}. Scheuermann et al.~\cite{scheuermann98vnv}
propose a way to approximate higher-order critical points using
Clifford algebra. Li et al.~\cite{li06rho} use interpolation schemes
based on a polar coordinate representation to detect vector field
singularities on a surface.

In contrast to the global variational approach taken in \cite{polthier03ivfs} in which the authors construct a discrete Hodge decomposition in order to obtain location and type of critical points of vector fields on polyhedral surfaces, our approach is local and cell-oriented. It may thus be easier to implement when a discrete parametrization of the surface is already given and can be used in conjunction with other local, grid-based methods.    

%-------------------------------------------------------------------

\section{Theoretical Background}
 \label{sec:theory}

The methods used for extracting vector field topology are founded
upon the theory of the qualitative behavior of dynamical systems. 
Most of the brief review of the essential theory in this section is
taken from the books \cite{dumortier06qtpds,book-perko91,book-andronov73}.

\subsection{Dynamical Systems}

\begin{definition}
\label{def:dynamicalsystem} A \emph{dynamical system} on
$E\subset\R[n]$, an open subset of $\R[n]$, is a $C^1$ map
$\vec{\phi}:\R\times E\ra E$, where $\vec{\phi}=\vec{\phi}(t,\vec{x})$ with
$t\in \R$, $\vec{x}\in E$, that satisfies
\begin{enumerate}
    \item $\vec{\phi}(0,\vec{x})=\vec{x}$ $\forall_{\vec{x}\in E}$,
    \item $\vec{\phi}(s,\vec{x})\circ\vec{\phi}(t,\vec{x})=\vec{\phi}(s+t,\vec{x})$ $\forall_{\vec{x}\in E,~u,v\in
    \R}$.
\end{enumerate}
\end{definition}
Dynamical systems are closely related to autonomous systems of differential equations. On the one hand, let $E\subset\R[n]$ be open and let 
$\vec{f}\in C^1(E)$ be Lipschitz continuous on $E$. Then the initial value problem of the autonomous system of differential equations
\begin{equation}
	\dot{\vec{x}}=\frac{d\vec{x}}{dt}=\vec{f}(\vec{x}),
 \label{eq:nlsysglobaldef}
\end{equation}
with $\vec{x}(0)=\vec{x}_0$ for any $\vec{x}_0\in E$ has a unique solution defined for all $t\in\R$ by virtue of the Picard–Lindelöf theorem. For each initial value $\vec{x}_0\in E$ this induces a mapping $\vec{\phi}(t,\vec{x_0}):\R\times E\ra E$ referred to as a \emph{trajectory} of the system (\ref{eq:nlsysglobaldef}). The mapping $\vec{\phi}$ then lies in $C^1(\R\times E)$ and thus is a dynamical system in the sense of Definition \ref{def:dynamicalsystem}. It is called the dynamical system induced by the system of differential equations  (\ref{eq:nlsysglobaldef}). On the other hand, if $\vec{\phi}(t,\vec{x})$ is a dynamical system on $E\subset\R[n]$, then
\begin{equation*}
	\vec{f}(\vec{x})=\frac{d}{dt}\vec{\phi}(t,\vec{x})|_{t=0}	
\end{equation*}
defines a $C^1$ vector field $f$ on $E$ and for each $\vec{x}_0\in E$, $\vec{\phi}(t,\vec{x}_0)$ is the solution to the initial value problem with $\vec{x}=\vec{x}_0$ of (\ref{eq:nlsysglobaldef}).

\subsection{Critical Points}
 \label{sec:criticalpointstheory}

An important concept in the field of dynamical systems is the notion of a critical point:
 
\begin{definition}
An \emph{equilibrium} or \emph{critical point} $\vec{x_0}\in\R[n]$ of
a dynamical system $\vec{\phi}$ is a point where
$\vec{\phi}(t,\vec{x_0})=\vec{x}_0$ $\forall_{t\in \R}$. If the dynamical
system is induced by a system of differential equations
(\ref{eq:nlsysglobaldef}), then a \emph{critical point} $\vec{x}_0$ of
$\vec{\phi}$ is a point where $\vec{f}(\vec{x_0})=\vec{0}$. If the Jacobian
$D\vec{f}(\vec{x_0})$ has only eigenvalues with nonvanishing real part,
$\vec{x_0}$ is called a \emph{hyperbolic} critical point. If $\det
D\vec{f}(\vec{x_0})\neq 0$, then $\vec{x_0}$ is called
\emph{non-degenerate} or \emph{first-order} critical point.
Otherwise it is called \emph{degenerate} or \emph{higher-order}
critical point.
\end{definition}

A system of differential equations can be approximated by its
linearization around a critical point $\vec{x_0}$ without changing its
qualitative behavior if $\vec{x_0}$ is a hyperbolic critical point of
that system (Hartman-Grobman theorem~\cite{book-perko91}). For
planar linear systems, only certain types of critical points can
occur and these can be classified in terms of the eigenvalues of the
Jacobian as shown in Fig.~\ref{fig:linearjacobianclassification}
(here only non-degenerate cases with $\det D\vec{f}(\vec{x_0}) \neq 0$ are
considered).

\begin{figure}
\centering
\includegraphics{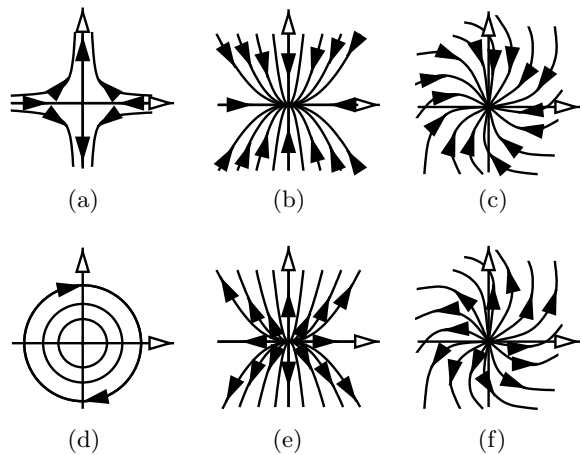}
\caption{First-order critical points of planar vector fields classified
  in terms of the eigenvalues $\lambda_1,\lambda_2$
  of the Jacobian: (a) saddle
  ($\Re(\lambda_1)\Re(\lambda_2)<0$, $\Im(\lambda_1)=\Im(\lambda_2)=0$),
  (b) attracting node ($\Re(\lambda_1),\Re(\lambda_2)<0$, $\Im(\lambda_1)=\Im(\lambda_2)=0$),
  (c) attracting focus ($\Re(\lambda_1),\Re(\lambda_2)<0$, $\Im(\lambda_1),\Im(\lambda_2)\neq0$),
  (d) center ($\Re(\lambda_1)=\Re(\lambda_2)=0$, $\Im(\lambda_1),\Im(\lambda_2)\neq0$),
  (e) repelling node ($\Re(\lambda_1),\Re(\lambda_2)>0$,
  $\Im(\lambda_1)=\Im(\lambda_2)=0$),
  (f) repelling focus ($\Re(\lambda_1),\Re(\lambda_2)>0$, $\Im(\lambda_1),\Im(\lambda_2)\neq0$).
  %A saddle critical point has a Poincar\'{e} index of $-1$ and a non-saddle critical point has index $+1$.
  }
\label{fig:linearjacobianclassification}
\end{figure}

\subsection{Poincar\'{e} Index of a Critical Point}
 \label{sec:indextheory}

In order to classify critical points of vector fields one can use the notion of the \emph{Poincar\'{e} index} (or \emph{index} for short) of a critical point.

\begin{definition}
Let $\vec{f}=(f_1,f_2)^T$ be a $C^1(E)$ vector field on some open
$E\subset\R[2]$. If $\vec{x_0}$ is an isolated critical point of
$\vec{f}$ and $\Gamma\subset E$ a Jordan curve such that $\vec{x_0}$
is the only critical point of $\vec{f}$ in its interior, then the
\emph{Poincar\'{e} index} of $\vec{x_0}$ (or \emph{index} for short)
is
$$I_{\vec{f}}(\vec{x_0}):= I_{\vec{f}}(\Gamma):=\frac{1}{2\pi}\varoint_{\Gamma}d\theta\in \Z,$$
with $\theta=\arctan \frac{f_2}{f_1}$.
\end{definition}

It can be shown that isolated first-order critical points (i.e.\
isolated critical points for which the Jacobian of the vector field
in the critical point has no eigenvalue of 0) have a Poincar\'{e}
index of $\pm$1 and that a saddle has a Poincar\'{e} index of $-$1,
whereas non-saddles have a Poincar\'{e} index of $+$1 (see Fig.
\ref{fig:linearjacobianclassification}).

\subsection{Topological Equivalence and Sectors}

Let us now establish the fundamental notion of topological equivalence of vector fields:

\begin{definition}
	Suppose that $\vec{f}\in C^1(E)$ and $\vec{g}\in C^1(F)$ with open sets $E,F\subset \R[n]$. The two autonomous systems of differential equations $\dot{\vec{x}}=\vec{f}(\vec{x})$ and $\dot{\vec{x}}=\vec{g}(\vec{x})$ and thier induced vector fields are said to be \emph{topologically equivalent} if there exists an orientation preserving homeomorphism that maps trajectories of the first system onto trajectories of the seconds system.  
\end{definition}

Markus \cite{markus54gso} showed that for planar $C^1$ systems of differential equations the condition of being topologically equivalent is equivalent to the systems having the same \emph{separatrix} configurations, where a separatrix of a system (\ref{eq:nlsysglobaldef}) is a trajectory of (\ref{eq:nlsysglobaldef}) which is either a critical point, a limit cycle, or a trajectory lying on the boundary of a \emph{hyperbolic sector} as defined below.  This justifies the use of the term \emph{vector field topology} for the topological skeleton of a vector field consisting of separatrices of that field. 

The notion of sectors was first introduced by Poincar\'{e}~\cite{poincare93msc} to investigate higher-order critical points of planar systems, and later extended by Ben\-dixon~\cite{bendixon01scd} and Andronov~\cite{book-andronov73}. The idea is that one can describe the qualitative behavior of a planar $C^1$ vector field $\vec{f}$ in a suitable neighborhood $N(\vec{x}_0)$ of an isolated critical point $\vec{x_0}$ of $\vec{f}$ in terms of connected regions, so called \emph{sectors}, which form a partition of $N(\vec{x}_0)$. Within each sector the trajectories of $f$ exhibit a behavior that is characteristic for this type of sector. It can be shown that there exist three topologically different types of sectors: 

\begin{definition}
A \emph{sector} of a critical point $\vec{x_0}$ can be classified as a
\emph{hyperbolic}, \emph{parabolic}, or \emph{elliptic} sector
according to its topological structure as shown in
Fig.~\ref{fig:topsectors}.
\end{definition}

\begin{figure}[tb]
\centering
\includegraphics{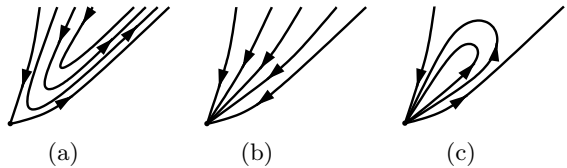}
\caption{The three distinct topological sectors of the vector field
         around
         an isolated critical point with nonvanishing Jacobian (modulo the operation of reversing the vector field direction):
         (a) hyperbolic sector (with two separatrices),
         (b) parabolic sector,
         (c) elliptic sector. }
\label{fig:topsectors}
\end{figure}

Tricoche et al.~\cite{tricoche00tsm} also used the idea of sectors to model higher-order critical points with a piecewise linear interpolation scheme.

\subsection{Bifurcation Theory}

Bifurcation theory is based on the notion of \emph{structural
stability} of a vector field due to Andronov and
Pontryagin~\cite{andronov37sg}: if the qualitative behavior of a
dynamical system (\ref{eq:nlsysglobaldef}) does not change for small
perturbations of the vector field $\vec{f}$, then that vector field is
called \emph{structurally stable}. If $\vec{f}$ is not structurally
stable, the topological skeleton of the vector field changes even
under small perturbations of the vector field $\vec{f}$ and $\vec{f}$ is
said to be \emph{structurally unstable} or to lie within the
\emph{bifurcation set}.

The perturbation of the vector field is usually modeled by an
additional parameter $\mu\in\R$:
\begin{equation}
\frac{d\vec{x}}{dt}=\vec{f}(\vec{x},\mu).
\label{eq:nlsysbifurcation}
\end{equation}
A value of $\mu=\mu_0\in\R$ for which the system
(\ref{eq:nlsysbifurcation}) lies in the bifurcation set is called a
\emph{bifurcation point} of (\ref{eq:nlsysbifurcation}) and $\mu_0$
is then called \emph{bifurcation value} of
(\ref{eq:nlsysbifurcation}).

Bifurcation theory has been studied extensively, see for example the
book by Guckenheimer and Holmes~\cite{book-guckenheimer90}. It also
explains the splitting of higher-order critical points into several
nearby first-order critical points: it can be shown that, if a
vector field $\vec{f}$ has an isolated critical point $\vec{x_0}$ of
higher order, there exists a perturbation of $\vec{f}$ such that
$\vec{x_0}$ splits into several isolated first-order critical points
nearby.

%-------------------------------------------------------------------

\section{Classifying Critical Points Without Using the Jacobian}
 \label{sec:classifying_points}

In this section, we show how the Poincar\'{e} index of an isolated
first-order critical point can be computed by evaluating the sign configuration
of the vector field's components on a finite set of sample points in
a neighborhood of the critical point, reminiscent of the
marching-cubes classification applied to isosurfaces in scalar
fields.

\subsection{Setting}
\label{sec:setting}

From now on let $\vec{f}\in C^1(E)$ be a vector field $\vec{f}:E\subset\R[2]\ra\R[2]$ defined on an open $E\subset\R[2]$ such that $\vec{f}$ is Lipschitz continuous on $E$ and only has non-degenerate first-order isolated critical points.

\subsection{$\omega$-level sets, Areas of Characteristic Behavior}
\label{sec:characteristicareas}

We introduce the notion of \emph{areas of characteristic behavior}
that will enable us to calculate the index of a critical point.

\begin{lemma}
Let $\vec{f}=(f_1,f_2)^T$ be a vector field like in \ref{sec:setting}
and $\vec{x_0}\in E$ an isolated first-order critical point of
$\vec{f}$, i.e. $\det D\vec{f}(\vec{x_0})$\\$\neq 0$. Then, $\vec{x_0}$ lies in
the intersection of the 0-level sets $c_1$ and $c_2$ of $f_1$ and
$f_2$. Furthermore, there exists an $\epsilon>0$ such that $c_1$ and
$c_2$ are infinitesimally straight lines (i.e. lie infinitesimally
close to straight lines) in an $\epsilon$-neighborhood of $\vec{x_0}$.
\end{lemma}

\begin{proof}
By definition, a critical point of $\vec{f}$ has to lie in the
intersection set of the 0-level sets of the components of $\vec{f}$.
Since $\vec{f}$ is linearizable around $\vec{x_0}$, Taylor's theorem
leads to
$\vec{f}(\vec{x})\approx\vec{f}(\vec{x_0})+D\vec{f}(\vec{x_0})(\vec{x}-\vec{x_0})$
with $\vec{x}$ in the $\epsilon$-ball $B_{\epsilon}(\vec{x_0})$. % for an $\epsilon>0$.
Since $D\vec{f}(\vec{x_0})$ has full rank, the 0-level sets of
$D\vec{f}(\vec{x_0})$ are straight lines intersecting in $\vec{x_0}$.
\end{proof}

\begin{definition}
\label{def:characteristicareas} Let $\vec{f}=(f_1,f_2)^T$ be a vector
field like in \ref{sec:setting} and $\vec{x_0}\in E$ an isolated
first-order critical point of $\vec{f}$. Then for an $\epsilon>0$ the
0-level sets $c_1$ of $f_1$ and $c_2$ of $f_2$ partition the
$\epsilon$-ball $B_{\epsilon}(\vec{x_0})$ of $\vec{x_0}$ into four
disjoint open subsets $A_1,\dots,A_4$ called \emph{areas of
characteristic behavior} or \emph{areas} for short.

In each area, the signs of $f_1$ and $f_2$ do not change, i.e. for arbitrary $1\leq i\leq 4$ and ${\vec{x},\vec{y}\in A_i}$ the following holds: 
\begin{equation*}
	f_j\neq0 \text{ and } \text{sgn}(f_j(\vec{x}))=\text{sgn}(f_j(\vec{y}))\text{ for $j=1,2$}.
\end{equation*}

The set of areas $A_1,\dots,A_4$ around a critical point can be seen
as an ordered, cyclic sequence of areas according to the order in
which they intersect with the boundary of $B_{\epsilon}(\vec{x_0})$,
as illustrated in Fig.~\ref{fig:areasequence}. We consider clockwise
traversal direction unless otherwise noted.

Since each of the two vector field components can either be positive
or negative in one area, there exist four different types of
characteristic areas as shown in Fig.~\ref{fig:areaturning}(a). In
the following, areas are written as ordered 2-tuples over the set
$\{+,-\}$ or equivalently as elements of the set $\{1,2,3,4\}$,
where $1=(+,+)$, $2=(+,-)$, $3=(-,+)$, $4=(-,-)$. Area sequences can
then be written as ordered 4-tuples over the set of areas, i.e.\ as
ordered 4-tuples over the set $\{1,\dots,4\}$. Two areas that lie
next to each other are called \emph{adjacent} areas. As $A_1$ and
$A_4$ are adjacent, the indices of the areas in the area sequence
are cyclic, i.e. the area sequence $A_1,\dots,A_4$ is glued together
at $A_1$ and $A_4$.
\end{definition}

\begin{remark}
For two adjacent areas $A_i,A_{i+1}$ of an area sequence with
$\vec{x}\in A_i$, $\vec{y}\in A_{i+1}$, either
\begin{align*}
  & \text{sgn}(f_1(\vec{x}))\neq\text{sgn}(f_1(\vec{y}))\wedge\text{sgn}(f_2(\vec{x}))
    = \text{sgn}(f_2(\vec{y})) \text{~~or} \\
  & \text{sgn}(f_1(\vec{x}))
    =
    \text{sgn}(f_1(\vec{y}))\wedge\text{sgn}(f_2(\vec{x}))\neq\text{sgn}(f_2(\vec{y}))
\end{align*}
holds, i.e.\ exactly one component flips its sign for two adjacent
areas but not both.
\end{remark}

\begin{figure}
\centering
\includegraphics{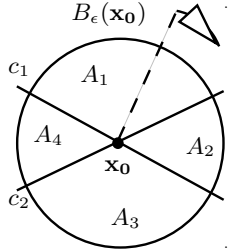}
\caption{Area sequence ($A_1,\dots,A_4$) around an isolated
         critical point of a $C^1$ vector field $\vec{f}=(f_1,f_2)^T$
         defined by the 0-level sets $c_1$ and $c_2$ of $f_1$ and $f_2$.
         The area sequence of a critical point $\vec{x}_0$ can be constructed by
         walking monotonously around the boundary of an $\epsilon$-ball
         $B_{\epsilon}(\vec{x}_0)$ of $\vec{x}_0$ starting at an arbitrary position
         (shown as dashed line above) and collecting the intersection points
         of $c_1$ and $c_2$ with $B_{\epsilon}(\vec{x}_0)$.}
\label{fig:areasequence}
\end{figure}

\subsection{Area Sequence and Types of Critical Points}
 \label{sec:areasequence}

\begin{figure}
\includegraphics[scale=0.85]{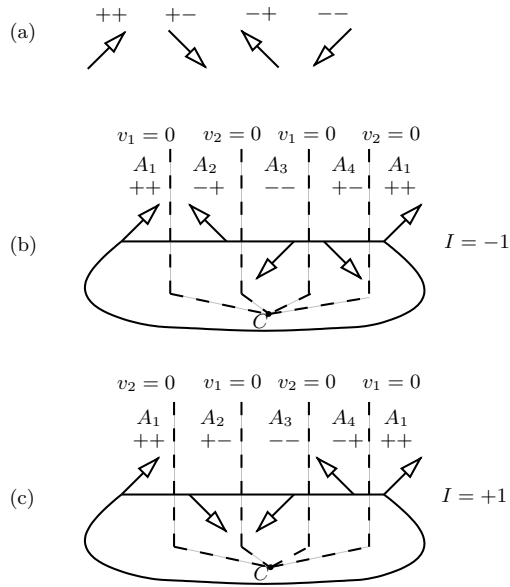}
\centering
 \caption{Areas of characteristic behavior and area sequences around a critical point $C$:
(a) the four types of characteristic areas $A_i$, (b)
counterclockwise and (c) clockwise turning behavior of the areas,
yielding a critical point of index $-$1 and $+$1, respectively.}
\label{fig:areaturning}
\end{figure}

\begin{definition}
 \label{def:areaturning}
Let $\vec{x_0}$ be an isolated first-order critical point of $\vec{f}$ defined like
above with an area sequence\\$(A_1, \dots, A_4)$. Then for a pair
$A_{i},A_{i+1}$ of adjacent characteristic areas of the area
sequence of $\vec{x_0}$ (see Fig.~\ref{fig:areaturning}(a)), a
\emph{turning} in the characteristic behavior of the vector field
can be defined. This turning can either be a \emph{clockwise
turning} or a \emph{counterclockwise turning} as defined in
Figs.~\ref{fig:areaturning}(b) and \ref{fig:areaturning}(c),
respectively.
Since Figs.~\ref{fig:areaturning}(b) and \ref{fig:areaturning}(c)
contain all 8 possible configurations of pairs of characteristic
areas, only a clockwise or counterclockwise turning behavior can
occur and the list in Figs.~\ref{fig:areaturning}(b) and
\ref{fig:areaturning}(c) is complete.
\end{definition}

\begin{remark}
One adjacent pair of characteristic areas already 
determines the whole area sequence in terms
of its turning behavior. This is the case because the vector field
components whose signs flip are alternating in the area sequence.
Thus, an area sequence can either show a clockwise turning behavior
or a counterclockwise turning behavior, but not a mixture of these
and we will also speak of a \emph{clockwise} or
\emph{counterclockwise turning behavior} of the area sequence as a
whole---a \emph{clockwise} or \emph{counterclockwise area sequence}
for short.
\end{remark}

We can use the turning behavior of area sequences around a critical
point to determine its Poincar\'{e} index:

\begin{theorem}
Let $\vec{f}\in C^1(E)$ be a vector field like in \ref{sec:setting} and let $\vec{x_0}\in E$ be an isolated critical point of $\vec{f}$ of first
order, such that $D\vec{f}(\vec{x_0})$ has full rank. If the area sequence of
$\vec{x_0}$ is counterclockwise, then the Poincar\'{e} index of
$\vec{x_0}$ is $I_{\vec{f}}(\vec{x_0})=-1$ and $\vec{x_0}$ is a topological
saddle of $\vec{f}$. If the area sequence of $\vec{x_0}$ is clockwise,
then $I_{\vec{f}}(\vec{x_0})=+1$ and $\vec{x_0}$ is a non-saddle
first-order critical point of $\vec{f}$.
\end{theorem}

\begin{proof}
The area sequence can also be interpreted as a sequence of four
qualitative samples of the vector field values lying on a piecewise
linear closed curve $\Gamma\subset B_{\epsilon}(\vec{x_0})$ that
contains $\vec{x_0}$ (qualitative in the sense that just the signs of
the vector field components are sampled). We will now prove that
summing up the angle change of the vector field between these
discrete samples and performing this for all four samples yields the
same result as evaluating the continuous integral of the change of
angle of $\vec{f}$ over a continuous Jordan curve $\hat{\Gamma}$
around $\vec{x_0}$, only containing one critical point $\vec{x_0}$.
Therefore, the Poincar\'{e} index of $\vec{x_0}$ can be computed by
just identifying the sequence of characteristic areas around
$\vec{x_0}$.

\begin{figure}
\includegraphics{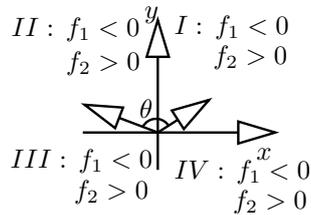}
\centering
 \caption{The change of angle of the vector field for two adjacent
          areas; since only one component flips sign, the two vectors lie in
          adjacent quadrants I--IV and one has $|\theta|<\pi$.}
\label{fig:areasvector}
\end{figure}

Let $\tilde{\vec{f}}$ be the linearized system of $\vec{f}$ defined on
an $\epsilon$-ball around $\vec{x_0}$. Since for two adjacent areas
only the sign of one component of $\vec{f}$ changes, the vector at
the first sampling point and the vector at the second sampling
point lie in two different but adjacent quadrants around $\vec{x_0}$
(see Fig.~\ref{fig:areasvector}). As the value of the linear
approximation $\tilde{\vec{f}}$ changes linearly on $\Gamma$ between
two sample points and $\tilde{\vec{f}}$ is continuous, the turning
behavior is uniquely determined by the two sample points; and
since the two sample points lie in adjacent quadrants defined by
the coordinate axes, the vector rotates about an angle $\theta$
with $|\theta|<\pi$ when walking on $\Gamma$ from one sample point
to the next. Thus, the whole area sequence (for four pairs of
areas) yields the total change in angle $\Theta$, with
$|\Theta|<4\pi$. As the area sequence wraps around the critical
point, the first and the last vector of our sample sequence are
the same and thus the change in angle of the vector field has to
be $2k\pi$ with $k\in\Z$, see Fig.~\ref{fig:areaturning}(b) and
Fig.~\ref{fig:areaturning}(c). As $\det D\vec{f}(\vec{x_0})\neq 0$,
for some $\epsilon$ also $\det D\vec{f}(\vec{x})\neq 0$ for $\vec{x}\in
B_{\epsilon}(\vec{x_0})$, i.e.\ the vector field is not constant in
a neighborhood of $\vec{x_0}$ and one has $k\neq 0$. It is $|k|<2$
because $|\Theta|<4\pi$. Thus, only $k=\pm 1$ is possible and the
critical point of the linearized field $\tilde{\vec{f}}$ is of
Poincar\'{e} index $I_{\tilde{\vec{f}}}(\vec{x_0})=\pm1$.

Since for linear systems all critical points of index $\pm1$ have
been classified (see Sections~\ref{sec:criticalpointstheory} and
\ref{sec:indextheory} as well as
Fig.~\ref{fig:linearjacobianclassification}), a critical point can
either be a saddle (for $I_{\tilde{\vec{f}}}(\vec{x_0})=-1$), if the
area sequence shows a counterclockwise turning behavior, or a
non-saddle (i.e.\ a source or sink for
$I_{\tilde{\vec{f}}}(\vec{x_0})=-1$), if the area sequence shows a
clockwise turning behavior. According to the Hartman-Grobman
theorem, the Poin\-car\'{e} indices of first-order isolated critical points are
invariant under linearization and $\tilde{\vec{f}}$ has a saddle in
$\vec{x_0}$ if and only if $\vec{f}$ has a topological saddle in
$\vec{x_0}$, which completes the proof of the theorem.
\end{proof}

\subsection{Invariant Operations on the Area Sequence}
 \label{sec:invariantoperationsareasequence}

There are 8 possible area sequences as there are four different
areas and each area sequence can be clockwise or counterclockwise.
Since these sequences yield only two types of critical points
distinguishable by their Poincar\'{e} index, it is desirable to
build equivalence classes of area sequences that yield critical
points of the same Poincar\'{e} index. These equivalence classes can
be directly constructed using elementary group theory. We refer to
the book~\cite{book-huppert67} for a comprehensive overview of the
theory of finite groups.

First it is obvious but nonetheless interesting to observe that
the turning behavior of the area sequence is invariant under
rotation and also invariant under a simultaneous sign flip of both
vector field components. Figure~\ref{fig:areaturning} illustrates
this: for example, the first and the last area pair of
Fig.~\ref{fig:areaturning}(b) and Fig.~\ref{fig:areaturning}(c)
can be related through a sign flip, likewise the second and third
area pair.

Rotations and sign flips can be modeled as group operations. Given a group $G$ and a set $X$, a (left) \emph{group action of $G$ on $X$} is a mapping $\circ:G\times X\ra X$ denoted $(g,x)\mapsto g\cdot x$ such that $e\cdot x=x$ for all $x\in X$ (here $e$ denotes the neutral element in $G$) and $(gh)\cdot x = g\cdot(h\cdot x)$ for all $g,h\in G$ and all $x\in X$.  The \emph{orbit} of an element $x\in X$ under a group action of
$G$ on $X$ is the set $O(x)=\{ g\cdot x : g\in G\}$. Orbits of a
group action on a set $X$ define equivalence classes on $X$, and
the set of all orbits is a partition of $X$.

The rotation of the area sequence is identical to the operation of the
cyclic group $C_4$ on the indices of the area sequence. The sign
flip can be modeled as the operation of a group isomorphic to the
symmetric group $S_2$ on the signs of the components. From now on,
the first group is referred to as the \emph{shape group} $G_s$ and
the second group as the (sign) \emph{flip group} $G_f$. The direct
product of $G_s$ and $G_f$ is called the \emph{coloring group} $G_c=G_s\times G_f$.

Now we define the group action of $G_c$ on the set of area
sequences. Let $\pi_r:G_c\ra G_s$ be the projection of $G_c$ onto
$G_s$ and $\pi_f:G_c\ra G_f$ the projection of $G_c$ onto $G_f$.
Further, let $g\in G_c$ be an element of the coloring group $G_c$
with its projections $\pi_r(g)=s\in G_s$ and $\pi_f(g)=f\in G_f$
onto the shape group and the flip group, respectively. Then, $g$
acts on an area sequence $(A_1,\dots,A_4)$ as defined below:
\begin{equation*}
g(A_1,\dots,A_4):=(fA_{s(1)},\dots,fA_{s(4)}),
\end{equation*}
where $s$ is a permutation of the indices $\{1,\dots,4\}$ and $f$
a self-inverse permutation on the set of areas, $\{1,\dots,4\}$.
$f$ can be interpreted as an element that flips the signs of all
components of the vector field: areas of type $(+,+)=\text{1}$ are
mapped to areas of type $(-,-)=\text{4}$ and vice versa; an area
of type 2 is interchanged with an area of type 3.

The orbits of this group action on the set of all possible area sequences yield equivalence classes of area sequences such that each equivalence class contains
all area sequences that can be mapped onto each other by rotations and sign flips.

\subsection{Interpolated Vector Fields}

Typical vector field data is given on a grid. Therefore, the
vector field needs to be interpolated to obtain values at non-grid
points. Local (i.e.\ cell-wise) interpolation schemes are often
chosen for the sake of simplicity and speed. In the following two
sections, we will have a closer look at two interpolation schemes:
barycentric interpolation on triangles and bilinear interpolation
on rectangles.

Both interpolation schemes share some important properties. On the one hand,
inside cells, both interpolation schemes are of class $C^{\infty}$
and thus linearizable. On the other hand, they are defined locally or in a
piecewise way: a different interpolant is used for each cell. The
interpolation is only of class $C^0$ across cell boundaries and it is
linear and continuous along cell boundaries. Basic concepts and tools are
developed for the simpler case of the barycentric interpolant.
Then, these tools are adapted and extended for bilinear
interpolation.

Note that the case of linear interpolation on triangular meshes can equally efficiently be solved by calculating the rotation along each triangle directly. As there can either be no or exactly one critical point of Poincar\'{e} index $\pm 1$ inside each cell, a cell with zero rotation has no critical point on its inside, whereas a nonzero rotation implies that there is a critical point inside the cell and already determines its Poincar\'{e} index. None the less we will describe the barycentric setting in the following as the notions introduced there will also be used in the more subtle case of the bilinear interpolant.

%-------------------------------------------------------------------

\section{Barycentric Case}
 \label{sec:barycentriccase}

A $d$-simplex $s$ on $d+1$ vertices offers a natural way for the
linear interpolation via barycentric coordinates. For the following,
we will again restrict the dimension to $d=2$, where a simplex is
identical with a triangle.

\subsection{Cells}

From now, a cell $T$ denotes a triangle with vertices
$\vec{x}_1$, $\vec{x}_2$, $\vec{x}_3\in\R[2]$ and with one real 2D vector
(from the vector field) attached to each vertex such that a cell
becomes a 3-tuple
$T=((\vec{x}_1,\vec{v}_1),(\vec{x}_2,\vec{v}_2),(\vec{x}_3,\vec{v}_3))\subset(\R[2]\times\R[2])^3$.
Most of the time, we are only interested in the vector field
values and not the position of the vertices, just writing
$T=(\vec{v}_1,\vec{v}_2,\vec{v}_3)$. The vector field inside a cell $T$
will be written $\vec{f}:|T|\ra\R[2]$, when $|T|\subset\R[2]$ is the
set of convex combinations of the vectors
$\vec{x}_1,\vec{x}_2,\vec{x}_3$. Further, let $\vec{f}=(f_1,f_2)^T$. We
restrict the vector field values on the vertices to be non-zero
such that the interpolated vector field has isolated singularities
only.

\subsection{Sector Sequences and level sets of the Interpolant}
\label{sec:barycentricsecseq}

The vector field is interpolated component-wise inside the cell
using barycentric coordinates. The interpolation is linear, and the
interpolant is continuous in the cell and on its boundaries.

\begin{definition}
\label{def:activeedgebarycentric} Given the cell $T$, let us look
only at one component of $\vec{f}$, say $f_1$ ($f_2$). If two
adjacent vertices (i.e. two endpoints of an edge of the cell) have
values of $f_1$ ($f_2$) such that one value is smaller than
$\omega$ and the other one bigger than $\omega$, then by virtue of
continuity and linearity of the interpolant on cell edges there
exists exactly one point on the edge with an $f_1$-value
($f_2$-value) of $\omega$. Such an edge is called an
\emph{$\omega$-active edge} for $f_1$ ($f_2$). If the two values
are both smaller or bigger than $\omega$, there exists no such
point and the edge is called an \emph{$\omega$-inactive edge} for
$f_1$ ($f_2$).
\end{definition}

\begin{remark}
We can now observe the following:

1) For one cell, there exist at most two $\omega$-active edges for each
       component and, thus, at most two points on the cell boundary with value $\omega$
       for each component if all vertex values of each component differ from $\omega$.

2) Since a critical point of the interpolated vector field can only
       occur in a crossing point of the 0-level sets of the components of
       the vector field, a cell with an interior critical point
       requires two active edges per component.
       Furthermore, there can at most be one critical point
       inside a cell because the 0-level sets of the interpolant are straight
       lines.

3) Since there can only be one critical point inside a cell, the sector
       sequence around a critical point can be found by looking for intersections
       of the 0-level sets of the components with the cell boundary.
       Each active edge yields exactly one 0-value
       point on a boundary edge, the position of which can be obtained
       through linear interpolation between adjacent vertices.
\end{remark}
If one collects the sequence of 0-value points of the components
while traversing the boundary of the cell (as shown in
Fig.~\ref{fig:walktriangle}), one can not only construct the area
sequence but also determine whether the 0-level sets cross and thus
whether there is a critical point of the interpolated vector
field. Let $a$ denote a 0-value point of the first component and
let $b$ denote a 0-value point of the  second component of the
vector field. Then, for the sequences $aabb$, $bbaa$, $abba$, or
$baab$, the cell does not contain a critical point, whereas for
the sequences $abab$ and $baba$, there is a critical point. These
$\binom{4}{2}=6$ intersection sequences are all possible sequences
for the barycentric interpolant because there can be at most two
0-value points for each component on the boundary of the cell.

The sequence of the crossings of the 0-level sets of the components
together with the information of how the signs of the components
change defines the sequence of characteristic areas around a
critical point and thus its Poincar\'{e} index as we have shown in
Section~\ref{sec:areasequence}. We will make use of this in the
following section.

\begin{figure}
\includegraphics{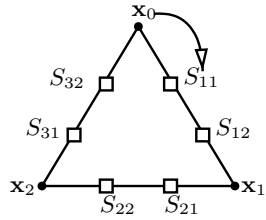}
\centering
  \caption{
  Triangle cell data structure with two edge slots on each edge.
  An edge slot is identified with a possible 0-value point of
  one of the vector field components.
  $S_{ij}$ refers to the $j$-th slot of edge $i$.}
\label{fig:walktriangle}
\end{figure}

\subsection{An Edge-Coloring Problem}
\label{sec:edgecoloringbarycentric}

Before we can classify critical points, a suitable data structure
for describing the area sequence around a critical point is
needed. Just looking at the sign configuration of the components
in the vertices (which can be seen as a coloring of the vertices)
does not allow us to distinguish the types of critical points
because the area sequence is not fully determined just by the
signs of the components in the vertices (see
Fig.~\ref{fig:vertexconfigtriangle}).

\begin{figure}
\centering
\begin{tabular}{@{}cc@{}}
\includegraphics[width=3cm]{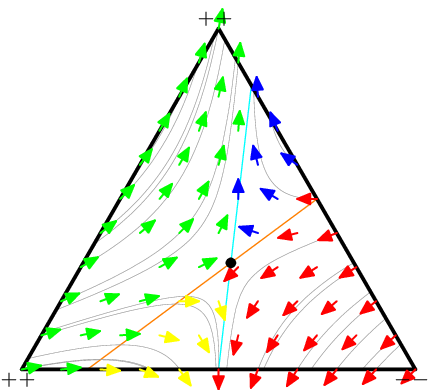}&\includegraphics[width=3cm]{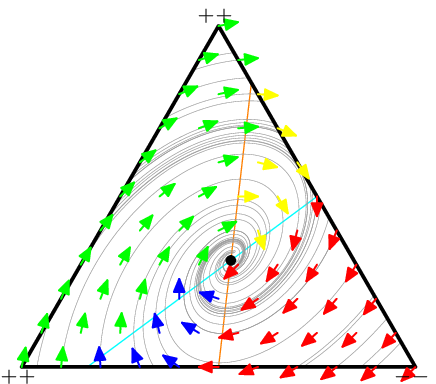}\\
(a)&(b)
\end{tabular}
 \caption{Different types of critical points for the same vertex sign configuration:
 (a) saddle with area sequence $(++,+-,--,-+)$,
 (b) attracting node with area sequence $(++,-+,--,+-)$.}
\label{fig:vertexconfigtriangle}
\end{figure}

We use an edge-based data structure to uniquely describe the
sequence of 0-values of the vector field components on the
boundary of the cell and the area sequence: each cell edge is
given two slots that can be filled with 0-value points of the
components (see Fig.~\ref{fig:walktriangle}). The slots represent
the order of the two components' zero values while walking around
the boundary of a cell. The slots also indicate whether a
component changes from positive to negative or from negative to
positive values as it passes through a 0-value point. There is no
more than one 0-value for each component along a boundary edge of
a cell because of the linearity of the interpolation along cell
edges. Hence each edge can be given two slots and 13 types of
boundary edges can occur, as listed in Table~\ref{tab:edgecolorings}. The notation
$(f_1:+\ra-, f_2:-\ra+)$ denotes that $f_1$ changes from a
positive value to a negative value, $f_2$ changes from a negative
value to a positive value, and the sign change of $f_1$ occurs
before the sign change of $f_2$ as one traverses the edge on the
cell boundary in clockwise direction. If a component is not
listed, it indicates that this component has no sign change on
that edge and thus no 0-value.

\begin{table}
\caption{The thirteen edge colorings.}
 \label{tab:edgecolorings}
\centering
\begin{tabular}{|c|l|}\hline
 Edge type&0-values of components\\\hline\hline
 $\xi_1$&no 0-value, neither for $f_1$ nor for $f_2$\\\hline
 $\xi_2$&$(f_1:+\ra-)$\\\hline
 $\xi_3$&$(f_1:-\ra+)$\\\hline
 $\xi_4$&$(f_2:+\ra-)$\\\hline
 $\xi_5$&$(f_2:-\ra+)$\\\hline
 $\xi_6$&$(f_1:+\ra-,f_2:+\ra-)$\\\hline
 $\xi_7$&$(f_1:+\ra-,f_2:-\ra+)$\\\hline
 $\xi_8$&$(f_1:-\ra+,f_2:+\ra-)$\\\hline
 $\xi_9$&$(f_1:-\ra+,f_2:-\ra+)$\\\hline
 $\xi_{10}$&$(f_2:+\ra-,f_1:+\ra-)$\\\hline
 $\xi_{11}$&$(f_2:+\ra-,f_1:-\ra+)$\\\hline
 $\xi_{12}$&$(f_2:-\ra+,f_1:+\ra-)$\\\hline
 $\xi_{13}$&$(f_2:-\ra+,f_1:-\ra+)$\\\hline
\end{tabular}
\end{table}

For a triangle cell, each of the three edges can now be of types
$\xi_1$--$\xi_{13}$, which represent a coloring of the triangle edges with 13
colors. Therefore, such an edge-coloring can be seen as an
ordered 3-tuple over the set of the 13 colors. Note that not every possible 3-tuple over the set of colors is a \emph{valid coloring} as there exist the following types of \emph{invalid colorings}, i.e.\ colorings that are not (or not uniquely) mappable to a sign configuration of the vector field at the cell vertices:

1) Colorings that have an invalid number of active edges:
       For each component, the cell has to have at least two active edges and the number of active edges has to be even.\\
       Example: the coloring
       $(\xi_2,\xi_4,\xi_1)$ would have one active edge for
       both $f_1$ and $f_2$ and would thus
       be an invalid coloring.

2) Colorings that are not mappable to a sign configuration of the vector
       field in the cell vertices at all.\\ Example: a component, say $f_1$, cannot
       change from $+$ to $-$ twice in a row. Thus, for example,
       a coloring $(\xi_2,\xi_2,\xi_2)$ would be invalid.

Each valid coloring of a triangle can
by construction be uniquely mapped to a sign configuration of the
values of the vector field components on the vertices (see
Fig.~\ref{fig:hierarchydata}). However, the edge-coloring of a triangle
carries more information than just the sign configuration of the
vector field components on the vertices (namely a unique
description of the intersection topology of the 0-level sets with
the cell boundaries) such that several colorings of the triangle
may be mapped to the same sign configuration.

\subsection{Classification}

As discussed in Section~\ref{sec:invariantoperationsareasequence},
there are certain operations that leave the clockwise and
counterclockwise turning behavior of the area sequence invariant.
These operations can be modeled as a group action of a certain
group on the set of area sequences.

Since there exists a bijective map from the set of edge-colorings
of a cell $T$ to the set of area sequences (each coloring
describes exactly one area sequence), one can make use of the
ideas presented in
Section~\ref{sec:invariantoperationsareasequence} to construct
equivalence classes of cell colorings yielding the same area
sequence. Then, all possible edge colorings of a cell can be
constructed and classified by using a group action that builds
equivalence classes of equivalent colorings. Colorings related by
rotation or color flip thus yield the same type of critical point.

\begin{theorem}
 \label{th:classificationbarycentric}
Every configuration of a cell $T$ (as a coloring of the cell) that
results in an intersection of the level sets of the two components
inside of $T$ and thus contains a critical point is topologically
equivalent to one of 8 representative configurations of which four
contain a saddle first-order critical point (Poincar\'{e} index
$-1$) and four contain a non-saddle first-order critical point
(Poincar\'{e} index $+1$). Representatives for the 8 orbits are
listed in Table~\ref{tab:casesbarycentric} and visualized (Online Resource 1). Colorings that are not included in this list
do not contain a critical point.
\end{theorem}

\begin{table}
 \caption{Representatives of the 8 equivalence classes of
 colorings of a triangle cell.}
 \centering
 \begin{tabular}{|c|l|r|}\hline
  Class&Cell coloring&Index\\\hline\hline
  1&$( 1, 6, 9 )$&$-1$\\\hline
  2&$( 1, 7, 8 )$&$+1$\\\hline
  3&$( 1, 10, 13 )$&$+1$\\\hline
  4&$( 1, 11, 12 )$&$-1$\\\hline
  5&$( 2, 4, 9 )$&$-1$\\\hline
  6&$( 2, 5, 8 )$&$+1$\\\hline
  7&$( 2, 11, 5 )$&$-1$\\\hline
  8&$( 2, 13, 4 )$&$+1$\\\hline
\end{tabular}
\label{tab:casesbarycentric}
\end{table}

\begin{proof}
This theorem is proven using a GAP program, see (Online Resource 3).

The basic idea is as follows. Using a combinatoric description,
all possible colorings of a cell are constructed: the set of all
ordered 3-tuples over the set of all possible colors. Then, a
group action on that set is considered to build equivalence
classes of colorings such that all pairs of elements of an equivalence
class can be mapped onto each other by means of
operations leaving the type of critical point invariant, as
described in Section~\ref{sec:invariantoperationsareasequence}.

Let $\mathcal{C}=\{\xi_1,\dots,\xi_{13}\}$ be the set of colors as described
above and let the set of 3-tuples over $\mathcal{C}$,
$\mathcal{T}=\{t=(t_1,\dots,t_3):t_i\in \mathcal{C}\}$ be the set of colored
cells or tiles. For the coloring group $G_c$, which is the direct
product of the shape group $G_s$ and the flip group $G_f$ (see
Section~\ref{sec:invariantoperationsareasequence}), an action of
$G_c$ on the set of all colored tiles $\mathcal{T}$ can be defined.
Here, the shape group $G_s$ is the cyclic group $C_3$ as the
subgroup of rotations of the symmetry group of the combinatoric
triangle (which is given by the symmetric group $S_3$). The color
flip group is chosen to be
$G_f=\langle (2,3)(4,5)(6,9)(7,8)(10,13)(11,12) \rangle$, which is isomorphic to
$S_2$; the numbers in the cycles correspond to the color indices of
edge colors. The choice of generator for $G_f$ is obvious and unique
since simultaneous flipping of signs of the two components results
in an interchange of edges with color $\xi_2$ and $\xi_3$, $\xi_4$
and $\xi_5$, etc.

Using our GAP program, first all
possible colorings of a cell are generated and for each orbit a
representative is checked for the validity of the coloring.  Elimination of invalidly colored cells can be done on a representative level: if
the representative of an orbit is an invalid coloring, all other
elements in the orbit are invalid because they are equivalent
colorings. Thus, one can talk of invalidly colored orbits.

After discarding all invalidly colored orbits, for the remaining
orbits the sequence of 0-value points of the components on the
cell boundary is extracted from the coloring.
Since the 0-level sets of the components are straight lines, the
sequence of 0-value points determine whether the 0-level sets
intersect within the cell (thus yielding a critical point) or not
(no critical point). If the 0-level sets intersect within the cell,
the area sequence and its turning behavior are extracted to
classify the type of critical point as saddle (Poincar\'{e} index
$-$1) or non-saddle (Poincar\'{e} index $+$1). As all critical
points are of first-order, this classification by the Poincar\'{e}
index covers all possible types of critical points. Since all
possible colorings have been considered the list of equivalence
classes is complete. See (Online Resource 1) for a list of all equivalence classes, including sample visualizations.
\end{proof}

%-------------------------------------------------------------------

\section{Bilinear Case}
 \label{sec:bilinear}

In this section, the concepts developed in the previous section
are transferred to the bilinear case. Most elements can be
immediately adopted, but some caution has to be taken because the
interpolant is no longer linear.

\subsection{Interpolation Scheme and Cells}

Bilinear interpolation is an extension of linear interpolation for
interpolating functions of two variables.
This interpolation scheme is simple, fast to implement, and widely
used in many visualization algorithms working in a cell-wise
fashion on uniform, rectilinear, or curvilinear 2D grids.

A cell $Q$ is represented by an ordered 4-tuple on the points
$\vec{x}_1,\dots,\vec{x}_4\in\R[2]$ with one real-valued 2D vector
(the vector field values) attached to each vertex such that a cell
becomes a 4-tuple
$Q=((\vec{x}_1,\vec{v}_1),\dots,(\vec{x}_4,\vec{v}_4))\subset(\R[2]\times\R[2])^4$.
Mostly, we are only interested in the vector field values, just
writing $Q=(\vec{v}_1,\dots,\vec{v}_4)$.

Any uniform, rectilinear, or convex curvilinear cell can be
transformed to the unit square $[0,1]^2$ in terms of a
diffeomorphism, yielding local Euclidean coordinates $s,t\in[0,1]$
within the cell. Without loss of generality we therefore only consider the case of Euclidean coordinates in the following, i.e. the setting when the interpolation scheme is employed in parameter space. We use the
following notation for bilinear interpolation:

Let $B_i:[0,1]\times[0,1]\ra\R$, $i=1,2$, be bilinear functions for the two
vector components (corresponding to $x$ and $y$), defined by
\begin{equation*}
 B_i(s,t):=(1-s,s)\mattwo{v_{1i}}{v_{2i}}{v_{3i}}{v_{4i}}\vectwo{1-t}{t},
\end{equation*}
where $s,t\in[0,1]\subset\R$ are local coordinates within one cell
and
$\vec{v}_1=({v_{11}},{v_{12}})^T$,\dots,$\vec{v}_4=({v_{41}},{v_{42}})^T$
are the values of the vector field components at the four vertices
$\vec{x}_1,\dots,\vec{x}_4$. Then $\vec{B}=(B_1,B_2)^T$ is called the
\emph{bilinear interpolant} of $\vec{v}_1,\dots,\vec{v}_4$.

\begin{remark}
 \label{rem:bilinearstandardform}
The bilinear functions can be rewritten as
\begin{equation}
B_i(s,t)=a_is+b_it+c_ist+d_i \label{eq:blstandard},
\end{equation}
with $a_i={(v_{3i} - v_{1i})}$, $b_i={(v_{2i}  -v_{1i})}$,\\$c_i={(v_{1i} - v_{2i} - v_{3i} + v_{4i})}$ and $d_i=v_{1i}$.

The following properties of the interpolant will be of importance
throughout the rest of the section:
\begin{enumerate}
 \item The $\omega$-level sets of $B_i$ are hyperbolas, which
       can be seen when writing the interpolants in
       standard form (\ref{eq:blstandard}) and interpreting
       them as conic sections. Some examples of $\omega$-level sets of $B_i$ are depicted in
       Fig.~\ref{fig:blinterpolantsscalar}.
 \item The interpolation is linear in each
       dimension and continuous. Furthermore, the interpolation along
       edges of the cell is linear: analogously to the
       barycentric case in Section~\ref{sec:barycentriccase}, one can
       define $\omega$-active and $\omega$-inactive cell edges.
       There exist four points on a cell boundary with value $\omega$ 
       and thus most four $\omega$-active edges, if all vertex values differ from $\omega$.
 \item In the case $c_i=0$, $B_i(s,t)$ is
       a linear function. If $c_i\neq 0$, $B_i(s,t)$ is
       nonlinear and when looking at the partial derivatives
       $\frac{\partial}{\partial s}B_i(s,t)=c_it+a_i$ and
       $\frac{\partial}{\partial t}B_i(s,t)=c_is+b_i$, one
       sees that $B_i(s,t)$ has an extremal point
       at $S_i(-\frac{b_i}{c_i}|-\frac{a_i}{c_i})$. This
       extremal point is a saddle because
       the Jacobian $DB_i$ is singular at $S_i$ and for all $s,t \in \R$ the Hessian $H_i$ of $B_i$ 
       has a negative determinant (and thus its eigenvalues are of opposite sign; see
       Figs.~\ref{fig:blinterpolantsscalar} and \ref{fig:blreduce} for examples).
\end{enumerate}
\end{remark}

\subsection{Bilinear Interpolation of Scalar Fields}

Let us look at the qualitative behavior of a bilinearly
interpolated scalar fields before extending this to vector fields.
We only consider one $B_i$ here, say $B_1$, as a placeholder for a
scalar field. The qualitative behavior of the interpolant within
the cell can be described in terms of its vertex configuration:

\begin{theorem}
 \label{th:classificationbilinearscalar}
Given $\omega\in\R$ and a cell $Q$ with vertex values $v_{11}$,
$v_{21}$, $v_{31}$ and $v_{41}$, $v_{i1}\neq\omega$ for $1\leq i\leq
4$, four cases for the qualitative behavior of $\omega$-level sets of
the interpolant $B_1(s,t)$ within a cell can occur:
Fig.~\ref{fig:blinterpolantsscalar} summarizes these cases that
depend on the classification of vertex values $v_{i1}$ being greater
or less than $\omega$. Here, active and inactive edges for a cell
$Q$ are defined in the same way as for the barycentric interpolant
in Definition~\ref{def:activeedgebarycentric}. Please note that we
use the terminology {\em saddle cell} (see
Fig.~\ref{fig:blinterpolantsscalar}) to describe a configuration of
vertex values $v_{i1}$, which is different from the classification
of a critical point as a saddle point. In the following, we denote
the first always by the compound term {\em saddle cell}.
\end{theorem}

All of the statements follow immediately from the continuity of the
interpolant and its linearity along cell boundaries.

\begin{figure}
 \centering
\includegraphics[width=2.5in]{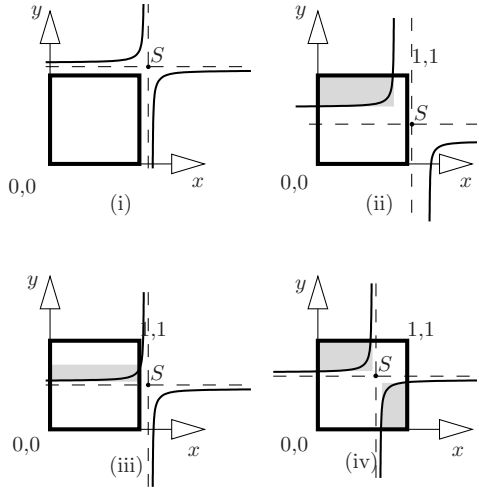}
 %\vspace*{-2ex}
 \caption{Examples of the four qualitatively distinct configurations of 0-level sets of bilinearly interpolated scalar fields: (i) inactive cell, (ii) single active cell, (iii) couple active cell, (iv) saddle cell. Bounding boxes (see Lemma \ref{lem:boundingbox}) are shown in gray.}
 \label{fig:blinterpolantsscalar}
\end{figure}

\begin{remark}
For an easier notation of the different cell types in terms of the
definitions of Theorem~\ref{th:classificationbilinearscalar}, we
write a cell $Q=(\vec{v}_1,\vec{v}_2,\vec{v}_3,\vec{v}_4)$ as a 4-tuple
over the set $\{+,-\}$ where the $i$-th entry of that tuple is set
to $+$ if $\vec{v}_i>\omega$ and to $-$ otherwise.
Table~\ref{tab:notationactivecellsbilinear} summarizes this tuple
notation, referred to as \emph{vertex sign configuration}, as in the barycentric case.
 \label{rem:notationactivecellsbilinear}
\end{remark}

\begin{table}
 \caption{Different cell types by vertex sign configuration.}
 \label{tab:notationactivecellsbilinear}
\centering
\begin{tabular}{|l|l|}\hline
\textbf{Cell type}&\textbf{Vertex sign configuration}\\\hline\hline
 inactive&$( -, -, -, - ), ( +, +, +, + )$\\\hline
 single active&$( -, -, -, + ), ( +, -, -, - ), ( +, +, +, - ),$\\
&$( -, +, -, - ), ( -, -, +, - ), ( -, +, +, + ),$\\
&$( +, -, +, + ), ( +, +, -, +)$\\\hline
 double active&$( -, -, +, + ), ( +, -, -, + ), ( +, +, -, - ),$\\
&$( -, +, +, - )$\\\hline
 saddle cell&$( -, +, -, + ), ( +, -, +, - )$\\\hline
\end{tabular}
\end{table}

\begin{lemma}[Bounding-box lemma]
\label{lem:boundingbox} Let $B_i(s,t)$ be a bilinear interpolant
on a rectangular grid cell $Q$ and let the cell be $0$-active. For
each curve of the $0$-level set connecting two $0$-value points
$A(A_x|A_y)$, $B(B_x|B_y)$ on the boundary of the cell, the
Cartesian product
\begin{equation*}
\begin{array}{@{}lll@{}}
L&=&[\min\{A_x,B_x\},\max\{A_x,B_x\}]\\
&&\times[\min\{A_y,B_y\},\max\{A_y,B_y\}]
\end{array}
\end{equation*}
defines a bounding box: the curve of the $0$-level set of $B_i$ is
inside $L$ except at the $0$-value points at the border of the cell. See Fig.~\ref{fig:blinterpolantsscalar} for an illustration of the set $L$ for some sample configurations.
\end{lemma}

\begin{proof}
We only give a sketch of the proof. The basic idea is the following.

Let the bounding box $L$ be given as the convex hull of the four points
\begin{equation*}
\begin{array}{l}
P_1(\min\{A_x,B_x\} | \min\{A_y,B_y\}),\\
P_2(\min\{A_x,B_x\} | \max\{A_y,B_y\}),\\
P_3(\max\{A_x,B_x\} | \max\{A_y,B_y\}),\\
P_4(\max\{A_x,B_x\} | \min\{A_y,B_y\}).
\end{array}
\end{equation*}
Note that $A,B\in L$. Then, the bilinear interpolant $B_i(s,t)$ restricted to the
bounding box $L$ can be expressed in terms of another bilinear
interpolant $\hat{B}_i(\hat{s},\hat{t})$ defined on $L$ and
interpolating the values of the points $P_1,\dots,P_4$. To leave
$L$, a 0-isoline of $\hat{B}_i$ has to cross the boundary of $L$.
Thus, a crossing 0-isoline can only occur where the value of
$\hat{B}_i$ on the boundary of $L$ is 0. One can then show that
apart from the two 0-values in the corner points $A$ and $B$ of a
$L$, there exist no other boundary points of $L$ with 0-value.
Therefore, such a crossing cannot occur and the 0-level set of
$\hat{B}_i$ completely lies inside $L$ and cannot leave $L$.

Since the interpolation function $\hat{B}_i$ along the boundary of
$L$ is continuous and linear, it is sufficient to show that only
two points in the set $L$ have 0-value. It then follows
immediately that these are the only points on the boundary with
zero value. This can be shown easily for a representative of each
of the cases listed in
Theorem~\ref{th:classificationbilinearscalar} and
Fig.~\ref{fig:blinterpolantsscalar}.
\end{proof}

\subsection{Analytical Computation of the Position of Critical Points}
 \label{sec:blinterpolant}

The position of critical points of the bilinear interpolant
$\vec{B}=(B_1,B_2)$ can be computed as intersection points of
0-level sets of $B_1$ and $B_2$ by solving
\begin{equation}
\vec{B}(s,t)=\vectwo{B_1(s,t)}{B_2(s,t)}=\vectwo{a_1s+b_1t+c_1st+d_1}{a_2s+b_2t+c_2st+d_2}=\vec{0}.
\label{eq:interpolantanalyticintersect}
\end{equation}
Since both $B_1(s,t)=0$ and $B_2(s,t)=0$ are linear or quadratic
equations in $s,t$, the system can either be combined into a
single linear or quadratic equation in $s$ or $t$. We will refer
to this equation as the \emph{intersection equation}. Since the
linear case occurring here can be treated in exactly the same way
as the barycentric case described in
Section~\ref{sec:barycentriccase}, this degenerate case will be
excluded in the following and we will only speak of the
(quadratic) intersection equation with discriminant $\Delta$. The
equation can either have no real solution ($\Delta<0$), two
different real solutions ($\Delta>0$), or a double real solution
($\Delta=0$), and the number of solutions of the intersection
equation gives the number of critical points, see
Fig.~\ref{fig:hyperbolatouching}.

Since the sign of the discriminant $\Delta$ determines whether the
vector field has no, one, or two critical points, it can be
interpreted as a bifurcation parameter of a dynamical system,
where $\Delta=0$ is the bifurcation value. Also $\Delta$ can be
interpreted as a measure for the stability of the critical
point---critical points with $\Delta=0$ are generally not stable,
whereas the case with two ($\Delta>0$) or no intersection points
($\Delta<0$) can be more or less stable according to the magnitude
of $|\Delta|$.

Given a rectangular cell with the vector field values\\
$\vv_1,\dots,\vv_4$ %where $\vv_i=(v_{i1},v_{i2})^T$ ($i=1\dots 4$),
one can write the components of the bilinear interpolant
%$\vec{B}(s,t)=(B_1(s,t),B_2(s,t))^T$ of the vector field in that cell
in normal form like done in (\ref{eq:interpolantanalyticintersect}) (see Remark~\ref{rem:bilinearstandardform}).

The intersection points of the 0-level sets of the components are
obtained by solving (\ref{eq:interpolantanalyticintersect}).
First, we determine which of the interpolants is linear in $s$ and
$t$ (i.e.\ for which $c_i=0$); if both are linear,
(\ref{eq:interpolantanalyticintersect}) becomes a linear system
and can be solved separately yielding one or no solution (and not
infinitely many solutions as the singularities are restricted to
be isolated). If one of the interpolants is linear and the other
nonlinear, the linear one can be solved for $s$ or $t$ and this
can be plugged into the other to yield a quadratic equation in $s$
or $t$. If both interpolants are nonlinear, one can be solved for
$s$ or $t$ and plugged into the other.

All these equations have at most two real solutions yielding one
coordinate of the 0-value points of the interpolated vector field.
The other coordinate can be obtained through the linear or
quadratic equation between $s$ and $t$ from the first step. A
solution within the cell is found by restricting the coordinates
to be in $[0,1]^2$. Note that the vector field values for each
component are always restricted to be nonzero at cell vertices and
that the two components are assumed not to be exactly the same.

\subsection{Level Sets and Sector Sequences}

Each solution of (\ref{eq:interpolantanalyticintersect}) lies in
the intersection set of the 0-level sets of $B_1$ and $B_2$. The
following cases can occur when looking at the intersections of
level sets of $B_1$ and $B_2$:

\begin{theorem}
\label{th:intersectionshyperbolas} Let $c_1$, $c_2$ be the
0-level sets of $B_1$, $B_2$ and let $\Delta\in\R$ be the
discriminant of the intersection equation for a cell $Q$. Then, the
following cases can occur:
\begin{enumerate}
 \item $\Delta<0$ $\Lra$ $c_1$, $c_2$ are disjoint and there is no
 critical point (see Fig.~\ref{fig:hyperbolatouching}(a)),
 \item $\Delta=0$ $\Lra$ $c_1$, $c_2$ touch in one point
 (see Fig.~\ref{fig:hyperbolatouching}(b)),
 \item $\Delta>0$ $\Lra$ $c_1$, $c_2$ intersect in two points
 (see Fig.~\ref{fig:hyperbolatouching}(c)).
\end{enumerate}
\end{theorem}
\begin{proof}
Clearly one has: $c_1$ and $c_2$ are disjoint $\Lra$ $\Delta<0$ which proves 1). The same holds for 2). Now
it remains to show 2), i.e.\ that $c_1$ and $c_2$ cannot touch twice and that the intersection equation has a double solution iff $c_1$ and $c_2$ touch in one point.

Suppose that $c_1$ and $c_2$ touch in one point as shown in
Fig.~\ref{fig:hyperbolatouching}(b).
A small perturbation of the vector field values now either results
in a case where the hyperbola branches do not intersect (i.e.\
$\Delta<0$) as shown in Fig.~\ref{fig:hyperbolatouching}(a), or it
yields a case where $c_1$ and $c_2$ intersect in two points, i.e.
where the intersection equation has two different solutions as shown
in Fig.~\ref{fig:hyperbolatouching}(c). As each pair of branches can
be perturbed separately a double touching case cannot exist as this
could be perturbed to yield a triple intersection of the branches
which is a contradiction to the maximum number of real solutions of
the intersection equation. This proves 3).
\end{proof}

\begin{figure}
\centering
\includegraphics[width=\linewidth]{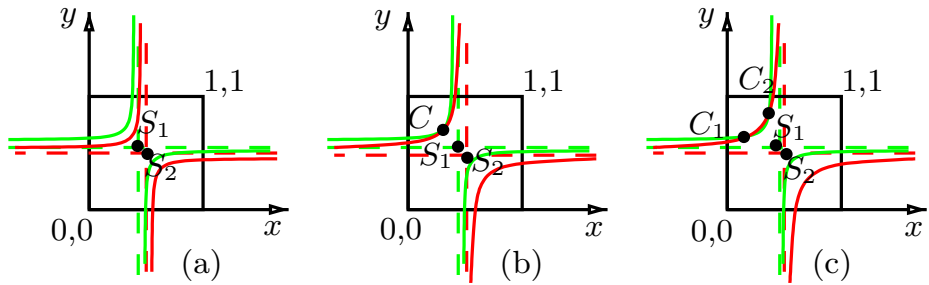}
\caption{
   Intersection cases for the 0-level sets (red and green lines) of the two vector field components: (a) no intersection with $\Delta<0$, (b) one touching point $C$ with $\Delta=0$, and (c) two intersection points $C_1$ and $C_2$ with $\Delta>0$. The saddle points $S_i$ $(i=1,2)$ of the two interpolants $B_i$ and their asymptotes (dashed lines) are depicted in red for $B_1$ and green for $B_2$.
} \label{fig:hyperbolatouching}
\end{figure}

\subsection{Critical Points and Area Sequences}

Let us now look at the types of critical points yielded for the
different intersection cases for the 0-level sets of the components
as investigated in Theorem~\ref{th:intersectionshyperbolas}.

1) The case $\Delta < 0$, where the 0-level sets do not intersect, is
trivial.

2) For the case $\Delta > 0$, when the 0-level sets of the components
intersect twice, yielding two first-order isolated critical points
of the interpolated vector field of first
order as none of the directional derivatives is $\vec{0}$ in the
intersection points. For each critical point separately, the same
things hold as for critical points in the barycentric case; in
particular we can look at area sequences, turning behavior of those,
etc. for each critical point separately.

\begin{remark}
The area sequences of the two
critical points always share three characteristic areas and these
are traversed in opposite directions for each of the points (i.e.\
clockwise turning for an area pair becomes a counterclockwise
turning and vice versa). Therefore, one of the two points has a
Poincar\'{e} index of $-1$ and the other one has a Poincar\'{e}
index of $+1$. Another way to see this is the fact that to first observe that there are at most two critical points inside a cell and that the Poincar\'{e}
index of a cell can either be $-1$, $0$ or $+1$. Now, by virtue of
the summation theorem for the Poincar\'{e} index, two critical points inside a cell of the same index would force the cell to have an index of $\pm 2$, a contradiction.
\end{remark}

3) For the touching case ($\Delta=0$) of the 0-level sets, the tools
developed in the previous sections cannot be applied: when
linearizing the field around the touching point, the 0-level sets of
the two components are tangential to each other and the Jacobian has
at least one zero eigenvalue. Thus, the critical point is not of
first-order and a well-defined area sequence does not exist for this
case. However, a perturbation of the field yields two critical
points, one of index $+1$ and one of index $-1$. Then, Poincar\'{e}
index theory states that the one critical point yielded by the
touching of the 0-level sets has an index equal to the sum of the
indices of the two critical points it splits into, i.e.\ the
second-order critical point from a touching of the 0-level sets has
index 0.

\subsection{Edge Coloring and Invariant Operations on the Colorings}

As the bilinear interpolant is linear along the cell edges, one can
employ the same ideas as in the barycentric case in order to
determine existence and types of critical points inside a cell,
again using an edge-based ``slot'' data structure. Using a data
structure with two slots for each edge of a cell, one can see this
as a cell coloring problem with a coloring of a cell represented as
a 4-tuple over the set of 13 colors like described in
Section~\ref{sec:edgecoloringbarycentric}.

Investigating the intersection topology of the 0-le\-vel sets becomes
more complicated because the 0-level sets are no longer straight
lines. Using the same notation as for the barycentric case, a
sequence of 0-value points $aabb$ (where again $a$ stands for a
0-value point of the first component of the vector field and $b$
for a 0-value of the second component) not necessarily has no
intersections. Here, it can either yield a case where the isolines
have no point in common, touch in one point, or intersect twice as
shown in Fig.~\ref{fig:intersectioncasesaabb}.

The question is whether for a sequence of 0-value points of the form
$\dots aa\dots$ or $\dots bb\dots$ the subsequence $aa$ or $bb$ can
yield a touching or intersection of the 0-level sets or not. If not,
the part $aa$ or $bb$ of the sequence is irrelevant for the
intersection topology and can be ''collapsed'', e.g.\ a sequence
$baababab$ for which $aa$ cannot yield a touching of the isolines
could be reduced to $bbabab$. If now $bb$ would yield no
intersection we could reduce once more to yield $abab$, a sequence
which cannot be collapsed any more.

For a fully reduced sequence, we can either determine the number of
crossings and thus critical points inside the cell as done in the
barycentric case (for sequences of the form $(ab)^n$, $n\in\N_0$) or
we know that the case is \emph{va\-lue-de\-pen\-dent}, i.e. the number of
intersections of the isolines cannot be determined just by looking
at the topology of the 0-value points on the boundary of the cell.
The latter is the case when the fully reduced sequence contains a
subsequence $aa$ or $bb$.

The question remains: how can one determine which subsequences of
type $aa$ or $bb$ can be collapsed? Let us first focus on non-saddle
cells. The bounding box lemma (Lem\-ma~\ref{lem:boundingbox}) gives us
valuable information for the case of single or double active cells.
It states that cases $aabb$ with single active cells for both
components lying at opposite vertices of one edge cannot intersect,
whereas all other configurations with single and double active cells
for the two components---although having no intersection in the
topology of 0-value points on the boundary---can be perturbed to
yield a touching or an intersection of the isolines and thus are
value-dependent.

\subsection{Saddle Cells}
 \label{sec:saddlecases}

There are several subtleties that have to be considered when dealing
with saddle cells. For one or both components of the vector field,
both branches of the 0-level sets intersect with the cell. The
topology of the 0-level sets inside the cell is not uniquely
determined by the sequence of 0-value points of the components on
the cell boundary. This ambiguity can be resolved by adopting the
asymptotic decider~\cite{gregory91tad}.

To determine the intersection topology for the 0-le\-vel\-sets inside
the cell, one can sort the 0-value points by their local coordinates
inside the cell in ascending (or descending) order. Then, for each
component, the first and second and the third and fourth point in
the sorted sequence of 0-value points belong pairwise to the same
branch of the 0-level set.

To model this behavior, the previously used edge-based data
structure is insufficient because the asymptotes imply a ``linking''
of two opposite edges of a cell. Therefore, the data structure needs
to be extended to describe the position of the 0-value points in the
sequence sorted by local coordinates. A natural choice for this is a
{\em double-edge} data structure shown in
Fig.~\ref{fig:hierarchydata}(a), which describes a cell by a pair of
double edges with four slots for each edge. When filling the slots
with 0-value points of the components, their position in the sorted
sequence for each coordinate is the index in the respective double
edge if one arranges the slots as in
Fig.~\ref{fig:hierarchydata}(a).

Each valid double-edge configuration can be uniquely mapped to a
valid edge coloring configuration via a natural projection $\pi_1$ (every valid double-edge configuration can be uniquely mapped to a edge coloring of a cell), which in turn can be mapped to a vertex sign configuration via a natural projection
$\pi_2$ (every valid edge coloring can be uniquely mapped to a vertex sign configuration of a cell), as described before for the barycentric case.
Figure~\ref{fig:hierarchydata} illustrates all three levels of
representation.

Using the double-edge data structure, one can now construct all
possible double-edge configurations for the case of a saddle cell
and check for the intersection topology of the 0-level sets of the
interpolant. As we will see in the following, the asymptotic decider
plays no role in the types of critical points yielded inside a cell,
i.e.\ there is no need to pay attention to how the 0-value point
sequence is reduced when determining the intersection topology and
area sequence for saddle cells.

\begin{figure}
 \centering
 \includegraphics[width=\linewidth]{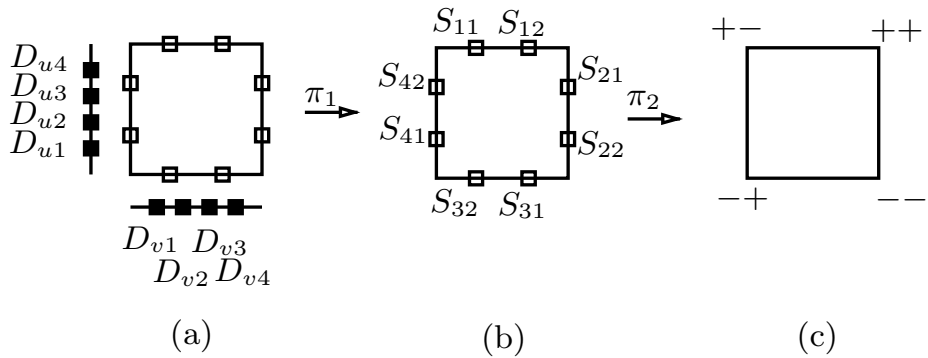}
 \vspace*{-4ex}
 \caption{Hierarchy of data structures to describe the
          sequence of 0-value points on the cell boundary from fine to
          coarse:
          (a) double-edge data structure consisting of two double edges
          with four slots, which can be projected via $\pi_1$ to (b) the edge data
          structure consisting of four edges with two slots each that
          can finally be projected via $\pi_2$ to (c) the vertex sign configuration of the components. }
          \label{fig:hierarchydata}
\end{figure}

An open question is how the $0$-value sequence can be reduced in a
valid way as it was done for non-saddle cells before. The tools
provided by the bounding box lemma (Lem\-ma~\ref{lem:boundingbox}) and
Theorem~\ref{th:intersectionshyperbolas} are instrumental here.
Figure~\ref{fig:blreduce} shows models of the two cases that contain
reducible subsequences of type $aa$ or $bb$ in their sequences of
0-value points on the boundary. All other configurations can yield a
touching or double intersection of the branches for the same
boundary topology and thus are value-dependent (by virtue of
Lem\-ma~\ref{lem:boundingbox} and
Theorem~\ref{th:intersectionshyperbolas}).

\begin{figure}
\includegraphics{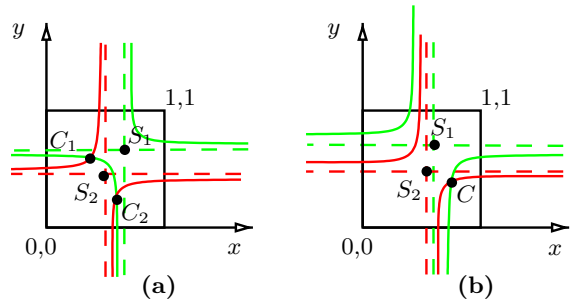}
 \centering
  \vspace*{-1ex}
 \caption{Examples of topologies of 0-value points of the
          components on the boundaries that can be reduced: (a)
          $baababba$,
          which can be reduced by collapsing the branch in the top right which
          cannot touch or intersect with other branches (bounding box lemma)
          to $bbabba$, (b) $abababba$, which can be reduced to $abab$ as the
          two branches in the top left cannot touch or intersect with others
          (bounding box lemma) or themselves (Theorem~\ref{th:intersectionshyperbolas}).
          Again, the saddle points $S_i$ $(i=1,2)$ of the two interpolants
          $B_i$ and their asymptotes (dashed lines) are depicted in
          red for $B_1$ and green for $B_2$.
} \label{fig:blreduce}
\end{figure}

Now that we have the tools to construct and reduce 0-value point
sequences for cells with bilinear interpolants, let us proceed in
the same way as we did for the barycentric case, classifying all
cell colorings in terms of the types of critical points they (may)
yield inside the cell.

\subsection{Classification}

\begin{figure}
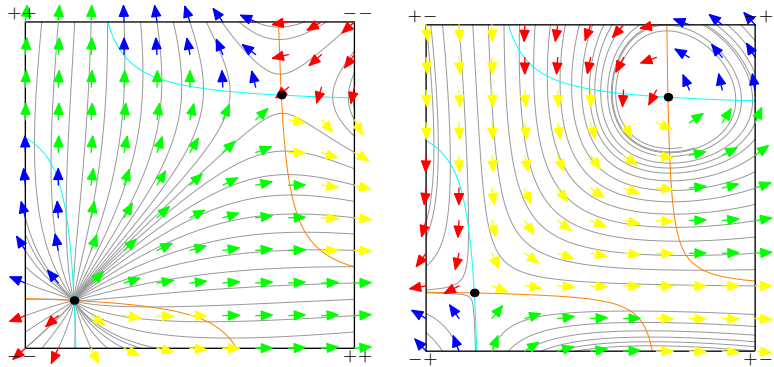
%[htp]
\centering
	
	\begin{tabular}{ll}
		\includegraphics[width=.4\linewidth]{rectcases.39}&
		\includegraphics[width=.4\linewidth]{rectcases.40}
	\end{tabular}

\caption{Two example visualizations of vector fields for which the $0$-value sequence on the cell boundary cannot be reduced, yielding two critical points inside the cell.}
\label{fig:rectdouble}
\end{figure}

\begin{theorem}
 \label{th:classificationbilinear}
Every configuration of a cell $Q$ (as a coloring of the cell) that
yields an intersection or touching of the level sets of the two
components inside of $Q$ is topologically equivalent to one of 74
representative configurations, see (Online Resource 2) for a
complete list:
\begin{itemize}
  \item 38 configurations have exactly one first-order critical point,
        the index of which is $\pm1$ and which is determined by the sequence
        of characteristic areas on the boundary of the cell.
  \item 4 configurations have two critical points, one of which is a saddle
        (i.e.\ index $-$1) and the other a non-saddle (i.e.\ index $+$1).
        Both critical points are stable. See Figure~\ref{fig:rectdouble} for visualizations of two sample cases.
  \item 32 configurations are value-dependent, i.e.\ they
        have either no, one, or two critical
        points. The case of a single critical point is an unstable
        bifurcation point of the underlying dynamical system.
\end{itemize}
Colorings that are not included in this list cannot have a critical
point inside the cell.
\end{theorem}

\begin{proof}
The proof is analogous to the proof of
Theorem~\ref{th:classificationbarycentric}. The corresponding GAP
program performs the following steps.

1) Create all possible cell colorings with the 13 colors and use a
group action of the coloring group $G_c=G_s\times G_f$ given by the
direct product of the shape group $G_s$ and the flip group $G_f$ on
the set of all colored cells to obtain equivalence classes of
colored cells (that contain the same number and types of critical
points). The flip group only depends on the colors and as these have
not changed it can be chosen as in the proof of
Theorem~\ref{th:classificationbarycentric}. The shape group $G_s$ is
chosen as the cyclic group $C_4$ of the rotations of a cell by
$\frac{\pi}{2}$ as a subgroup of the full symmetry group of the
square, which is the dihedral group $D_8$ of order $8$.

2) Examine the representatives of the orbits: sort out invalid
colorings as in the proof of
Theorem~\ref{th:classificationbarycentric} and determine the
intersection topology and the reduced 0-value sequence of the
0-level sets of the two components using the properties described in
the preceding sections. This yields an area sequence that can be
used to determine the number and types of critical points inside the
cell. Here, four classes of configurations are distinguished: cases
with exactly one critical point, cases with exactly two critical
points, cases that are value-dependent, and cases that yield saddle
cells for any component.

3) Determine the turning behavior of the reduced area sequence and thus
the type of critical point for configurations with exactly one
critical point. Va\-lue-de\-pen\-dent and double intersection cases need
no further treatment as the number and types of the critical points
are not fully determined by the area sequence.

4) For saddle cells: construct all possible double-edge configurations
and the reduced 0-value sequence for each configuration yielding an
area sequence.

As it turns out, for saddle cells all double-edge configurations
yield the same number and the same types of critical points such
that the asymptotic decider indeed does not influence the types of
critical points inside a cell as claimed in
Section~\ref{sec:saddlecases}. For the other cases, the number of
equivalence classes given in the theorem arise. A full list of
cases, including example diagrams, is provided in (Online Resource 2).
\end{proof}

%-------------------------------------------------------------------
\section{Boundary Points and Closed Streamlines}
 \label{sec:boundarypoints}

So far the issue of critical points on cell boundaries has not been
addressed and is somewhat delicate because the interpolation scheme
is only $C^0$ across cell boundaries for both the barycentric and
the bilinear cases.

To resolve this issue, we employ a cell clustering as described by
Tricoche et al.~\cite{tricoche00tsm}: the current cell with a
critical point at the cell boundary and its neighboring cell sharing
the edge with the critical point are clustered into a {\em
super cell}, see Fig.~\ref{fig:bdpoints}. This process is iterated as long as there are critical points on the cell (or super cell) boundary. Then, a new (artificial) critical point $C'$ is positioned inside the super cell at a position given by the averaged positions of the critical points inside the super cell. This new critical point $C'$ can be of arbitrary order and complexity.
In order to determine the topology of the vector field inside the super cell, the super cell is triangulated with a new inner vertex at the position of $C'$. Subsequently, a piecewise linear vector field is constructed inside the super cell as the union of the barycentrically interpolated vector fields on the triangles, using the information from the boundary of the super cell. As the direction of the new vector field does not change on rays emerging from $C'$, hyperbolic sectors around $C'$ can be identified by looking at the configuration of the vector field on the boundary of the super cell. After identification of the boundaries of hyperbolic sectors, the separatrices can be traced beginning from the super cell boundary in the usual way. Figure~\ref{fig:sectorssuper} shows how the different sector types around the critical point $C'$ can be determined by using a sequence of orientations of the vector field in relation to the position vector on the boundary of the super cell with respect to the origin $C'$. We refer to
\cite{tricoche00tsm} for details of the method. This is a local
method that does not alter the vector field behavior outside the
super cell; thus, it is a robust and simple way to deal with critical
points on cell boundaries. Note that the data of the randomly generated test vector fields as shown in Table~\ref{tab:dataresults} suggest that boundary critical points occur extremely seldom such that the additional steps needed by the algorithm in that case are negligible regarding its speed.

Another problem not addressed so far is the task of finding limit cycles of vector fields. Since this problem has been addressed in various publications already (see \cite{wischgoll01dvcs,theisel04gidcsl}) and the methods presented there can be incorporated in our algorithm, we will not deal with this case here.

\begin{figure}%[htp]
\centering
	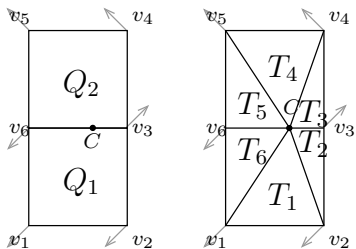
	\caption{Forming a super cell consisting of the quadrangular cells $Q_1$ and $Q_2$ for a case where a critical point $C$ lies on the cell boundary of $Q_1$ and $Q_2$. The super cell is then triangulated by $T_1,\dots,T_6$ and the vector field (shown in gray) is interpolated barycentrically inside each triangle $T_i$.}
	\label{fig:bdpoints}
\end{figure}

\begin{figure}%[htp]
\centering
	\includegraphics[width=0.5\linewidth]{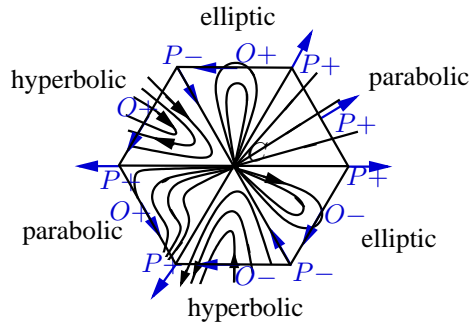}
	\caption{Possible configurations of the orientation of the vector field in relation to the position vector with respect to the origin $C'$. $P$ denotes a parallel configuration, $O$ an orthogonal one. The appended signs $+$ and $-$ distinguish the two possible directions of the vector field w.r.t. the position vector in each case.}
	\label{fig:sectorssuper}
\end{figure}

%-------------------------------------------------------------------
\section{Algorithm}
 \label{sec:algorithm}

The algorithm consists of three main stages. First, all cells that
may contain critical points are identified. Then, these cells are
examined more thoroughly to check whether they contain one or more
critical points and whether cell clustering is needed. This step is
performed by computing the coloring of the cell and using a lookup
table to determine the class of the cell coloring as defined in
Theorems~\ref{th:classificationbarycentric} and
\ref{th:classificationbilinear}. Finally, the positions of the
critical points inside the cells are computed analytically. For
non-saddles, the type can either be classified using the Jacobian at
that point or the sector-based idea described in
\cite{tricoche00tsm}, which is also used to determine the types of
higher-order critical points. Note that the exact type of
first-order non-saddles is not required to build the topological
skeleton, which is only based on trajectories originating from
hyperbolic sectors.

\subsection{Generic Algorithmic Skeleton}
 \label{sec:algorithm_skeleton}

\begin{enumerate}
\item
Traverse cells, mark cells that have at least two active edges for
each component by looking at the sign configuration of the vector
field at the vertices.

\item
Compute cell colorings of the marked cells and fetch the
configuration class for each cell in the lookup table; a perfect
hash function can be used for this. Deal with critical points on the
boundary of the cell using the cell-clustering approach as described
in Section~\ref{sec:boundarypoints}.

\item
Compute the location(s) of the critical point(s) inside the cell.
For the barycentric case, a linear system is solved. For the
bilinear case, one can write the interpolant for each component in
normal form $B_i(s,t)=a_is+b_it+c_ist+d_i$ and combine the two
equations into the linear or quadratic intersection equation in $s$
or $t$. Then the discriminant $\Delta$ of the intersection equation
and its solutions are computed. For both cases, solutions lying
outside the cell are dropped.

\item
Construct the topological skeleton of the vector field in the usual
way tracing streamlines at the boundary of hyperbolic sectors and
connecting these to other critical points or boundary points when
they get close. The boundaries of hyperbolic sectors can either be
found as described by Tricoche et al.~\cite{tricoche00tsm} or by
using the eigenvectors of the Jacobian at the critical point. Our
implementation employs an adaptive explicit Runge-Kutta method 
for streamline tracing.
\end{enumerate}

\subsection{Performance and Topology Simplification}
\label{sec:performance}

We propose that, in order to overcome numeric problems in the
vicinity of critical points, streamlines are only traced from and to
boundaries of cells that contain critical points. Within these
cells, streamlines are approximated by a straight line connecting
the intersection point of the trajectory and the cell boundary with
the critical point. This approach has the advantage that typically
many small steps of an explicit solver in the vicinity of the
critical point can be saved, thus speeding up the calculation of the
topological skeleton.

Further simplifications leading to speed-ups can be considered.
First, for small cells, the location of the critical point can be
reasonably well approximated by placing the critical point at an
arbitrary position---say, the center---of the cell. Second, the
topology of the vector field for cells with two critical points can
be simplified by replacing the two first-order critical points by a
single artificial second-order critical point.\\The position of the
second-order critical point can be approximated by the midpoint of
the straight line connecting the two first-order critical points.

\subsection{Numerical Stability}

To be numerically stable an algorithm has to yield consistent
results for the same input data, no matter how the input data is
ordered. This is especially important for floating-point data, where
the result of interpolation may depend on the interpolation
direction and order of instructions. For the algorithm in
Section~\ref{sec:algorithm_skeleton}, this means that no matter how
a cell is oriented in terms of the geometric location of its
vertices, the result of the calculations has to be invariant for the
two cells.  This can be achieved by choosing a unique interpolation
direction along cell edges such that this choice always yields the
same sequence of calculations on the floating-point vector field
data. Then, by virtue of consistency of IEEE floating-point
arithmetic~\cite{ieee85sbf}, the interpolation result and the
corresponding vertex and edge classifications have to be identical.

Consider a cell with vertices $P_1,\dots,P_n$, $P_i(x_i|y_i)$\\
($i=1\dots n$), where $x_i$ and $y_i$ are floating-point numbers.\\
Then, without loss of generality one can choose the interpolation
direction along an edge between two vertices to always point from
the vertex with smaller $y$-coordinate to the vertex with greater
$y$-coordinate. If the two $y$-coordinates agree, then the same
argument can be applied to the $x$-co\-or\-di\-na\-tes of the vertices. This
method always yields a unique interpolation direction as: (1) the
inequality comparison operator for IEEE floating-point numbers is
consistent in the sense that if $f_1$ and $f_2$ are two
floating-point numbers, $f_1>f_2\Ra f_2<f_1$; (2) the equality
comparison operator for IEEE floating-point numbers is commutative,
i.e.\\$f_1=f_2\Ra f_2=f_1$.

Our algorithm and the classification of the critical points depend
on a consistent orientation of traversing the boundary of a cell.
This orientation can be defined by the sign of the oriented area
of a cell. For a
triangle with vertices $P_1(x_1|y_1)$, $P_2(x_2|y_2)$,
$P_3(x_3|y_3)$, the oriented area is
$$A^{*}=\frac{1}{2}(x_2-x_1)(y_3-y_1)-(y_2-y_1)(x_3-x_1).$$ For the
case of quadrilateral cells, a cell is split into two triangles
and the oriented area sequence of one of the two triangles can be
used.

%-------------------------------------------------------------------
\section{Results}
 \label{sec:results}

\begin{figure}[h]
\includegraphics[width=0.49\linewidth]{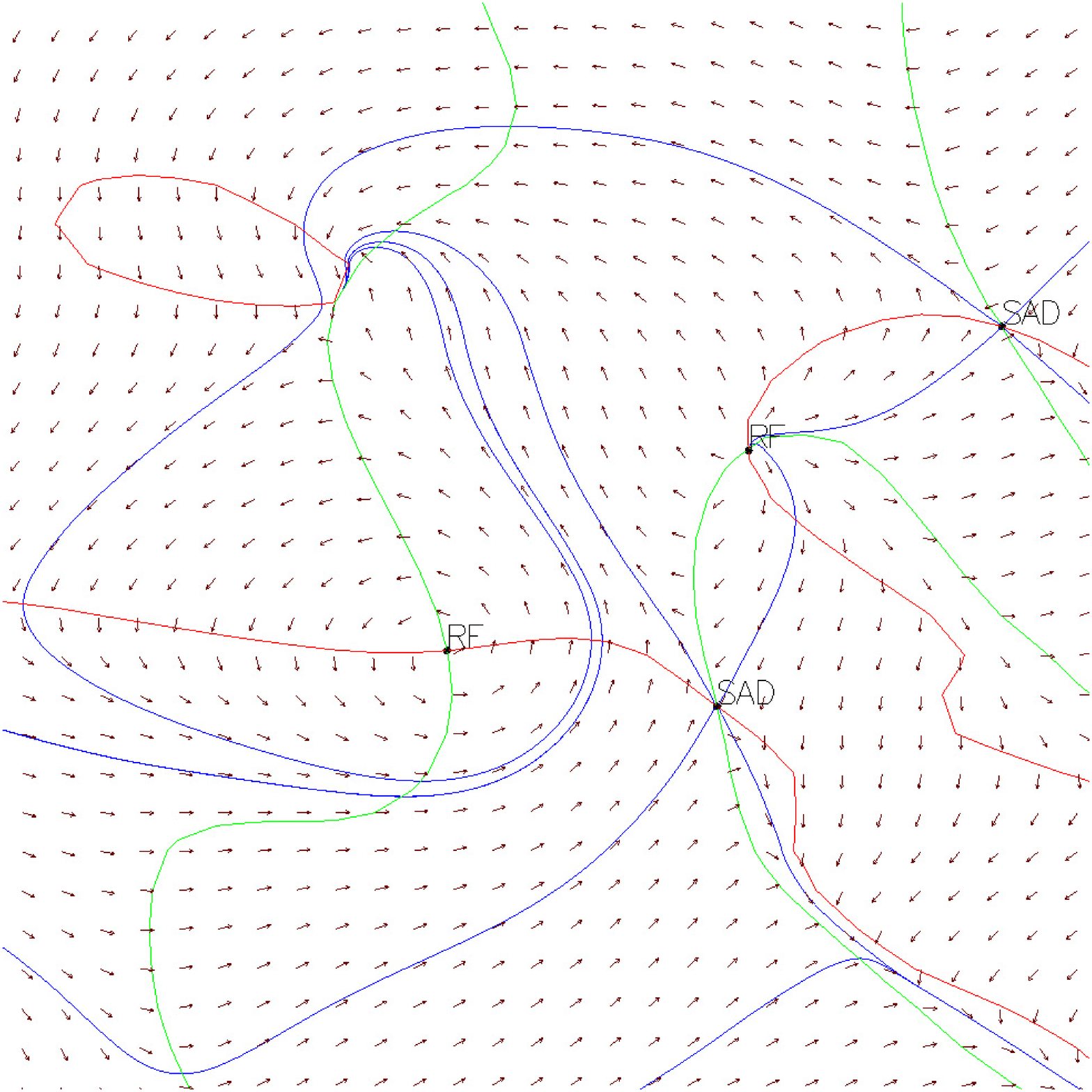}
\hfill
\includegraphics[width=0.49\linewidth]{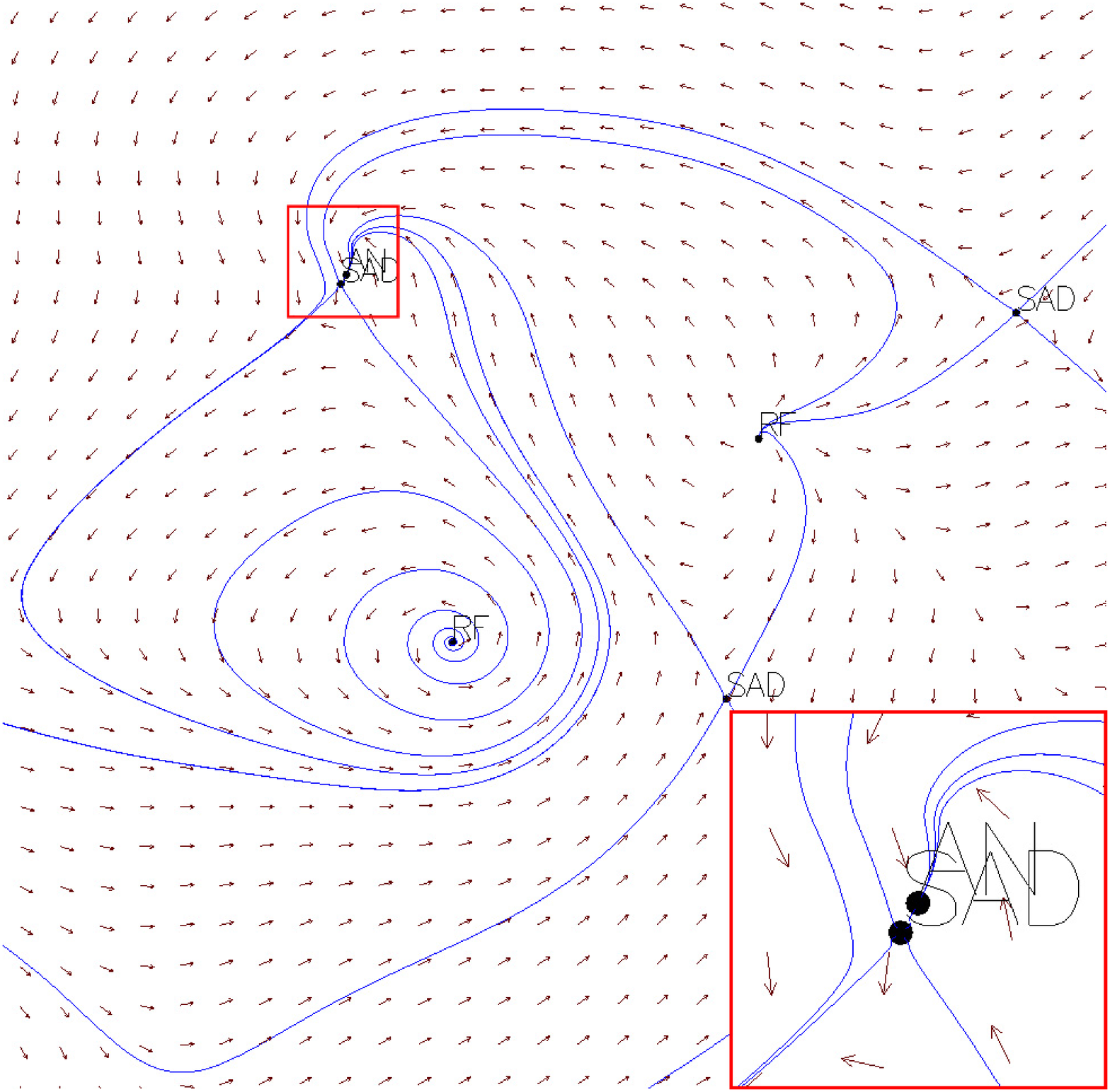}\\
\mbox{}\hfill(a) \hfill\hfill (b) \hfill\mbox{}
 \caption{Vector field
topology of the oceanic flow data set. Black dots indicate detected
critical points, blue lines show the topological skeleton,
annotations provide the classification (AN: attracting node, RF:
repelling focus, SAD: saddle). (a) Incorrect topology obtained by a
method that intersects linearized 0-level sets of the $x$-component
(red) and the $y$-component (green). Notice how trajectories
terminate in critical points that are not detected as such in the
top left of the image. (b) Correct topology obtained with our
method. The pair of critical points inside a value-dependent cell is
correctly identified (framed area shown magnified in the bottom
right).} \label{fig:topologyocean}
\end{figure}

\begin{figure}[h]
\centering
\includegraphics[width=0.7\linewidth]{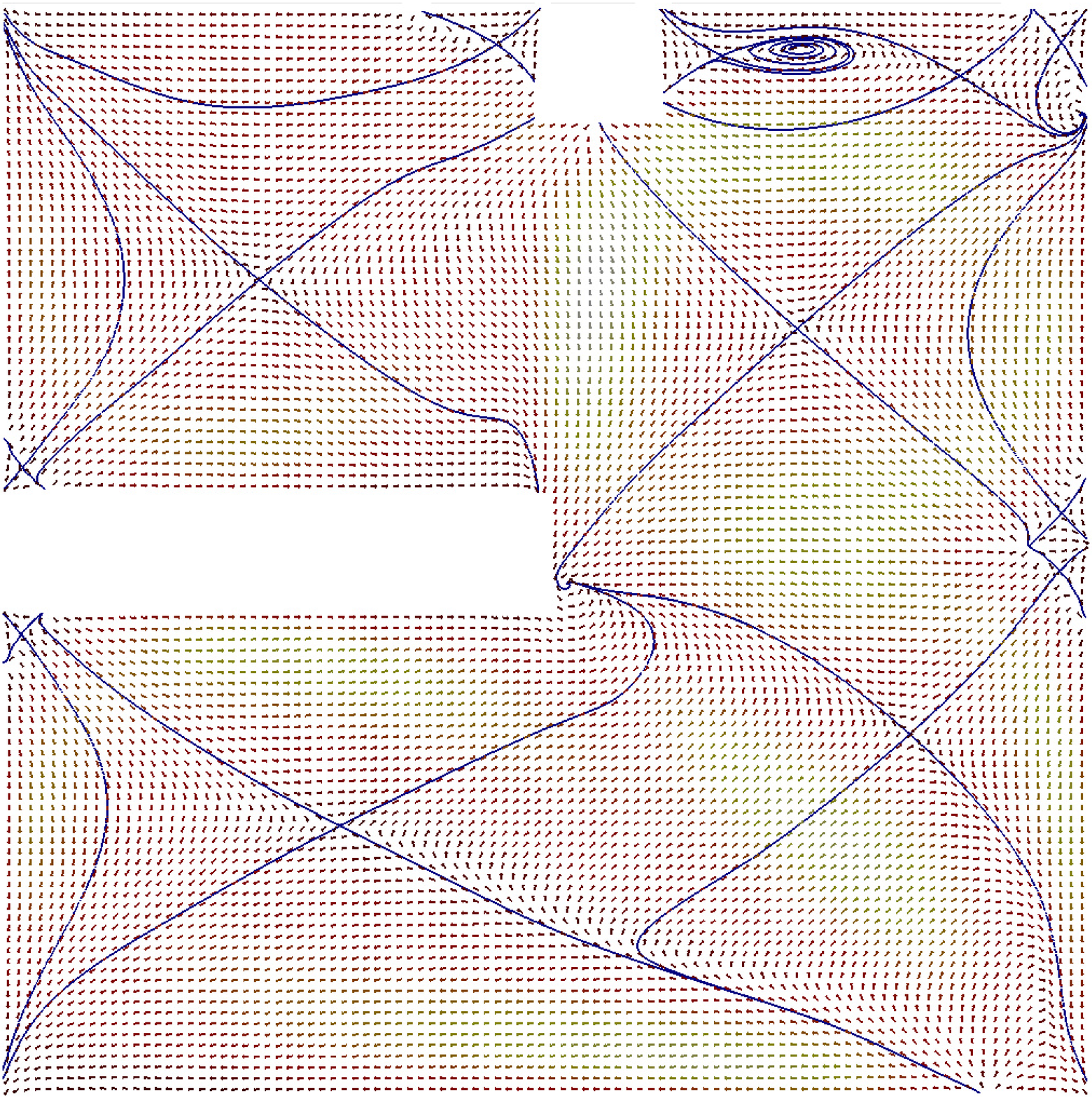}
\caption{Vector field topology of the air flow data set calculated with our algorithm. Separatrices are shown in blue along with a glyph-based visualization of the vector field, where the arrows encode the vector direction and the vector norm is color-coded.}
\label{fig:topologyair}
\end{figure}

\begin{figure}[h]
\centering
\begin{tabular}{cc}
\includegraphics[width=.45\linewidth]{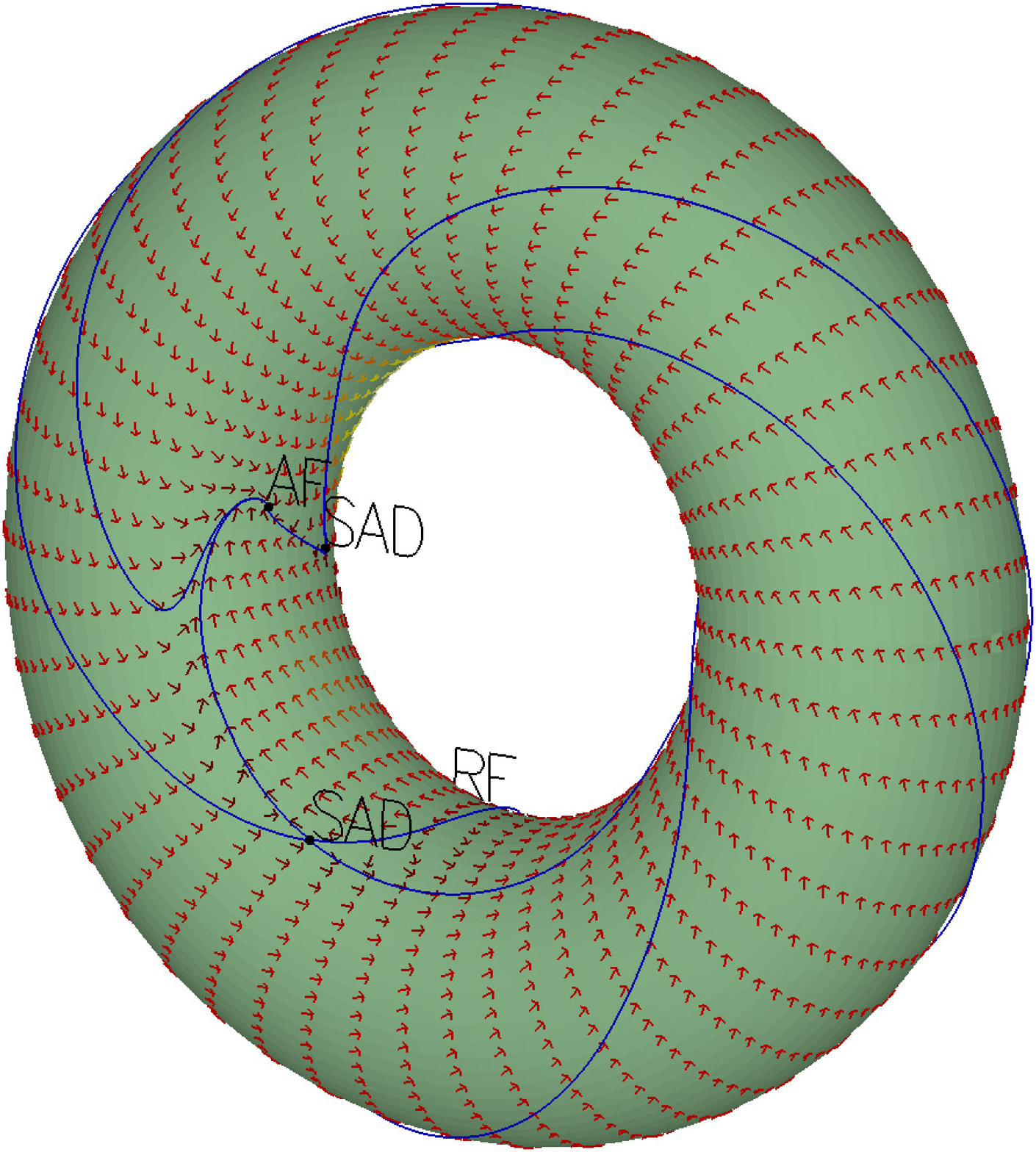}&\includegraphics[width=.45\linewidth]{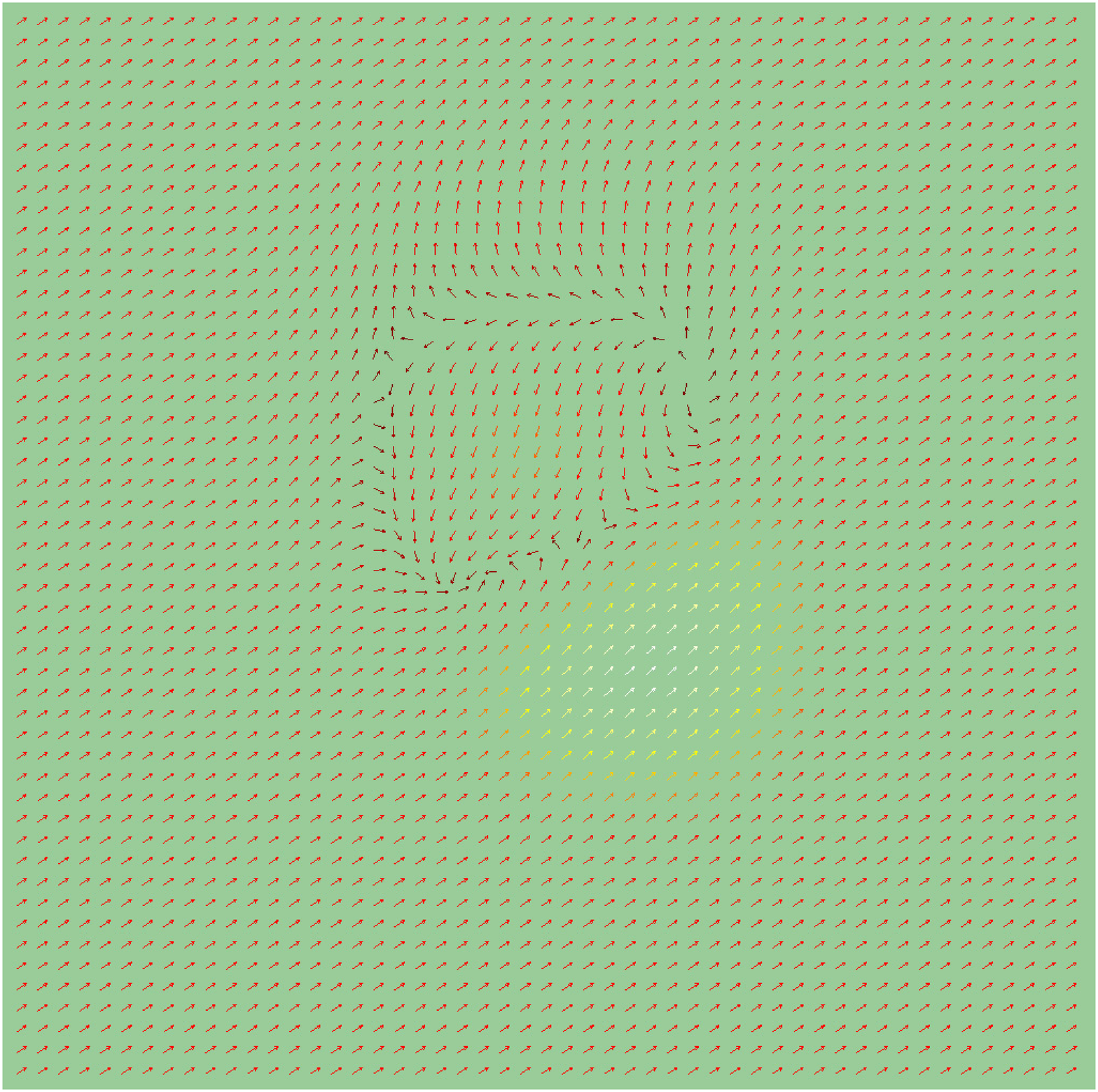}
\end{tabular}
\caption{Example of a vector field on a torus. Left: torus immersed in Euclidean 3-space, right: torus in parameter space. The top and the bottom as well as the left and the right border are identified, respectively. Again, the vector field is also visualized via glyphs.}
\label{fig:torusvf}
\end{figure}
 
We illustrate the results of our algorithm for four test cases. The
first test case consists of randomly generated bilinearly
interpolated vector fields: for each cell vertex, the two vector
field components are random numbers lying in
$[\text{-1,1}]\backslash\{\text{0}\}$, quantized to
$\text{10}^\text{-6}$. Machine accuracy is set to
$\epsilon=\text{10}^\text{-9}$ and the threshold of a critical point
being of second-order is chosen to be $\Delta < \text{0.05}$, where
$\Delta$ is the discriminant of the intersection equation as
described in Section~\ref{sec:blinterpolant}. We have generated and
examined $10^8$ cases. The second test case is a CFD dataset of an
oceanic flow in the Baltic Sea (simulation courtesy of Kurt
Frischmuth, University of Rostock), see Figure~\ref{fig:topologyocean}. The data set is given on a uniform grid of size 100\;$\times$\;112, and the 
vector field is rotated by 90 degrees, c.f. Theisel et al.~\cite{theisel03c3v}.

The third data set examined with the algorithm is a planar slice of a simulated unsteady flow of air in a hot room at a fixed time step of the simulation (data courtesy of Filip Sadlo, Universität Stuttgart). See Figure~\ref{fig:topologyair} for an illustration of the flow topology extracted with our algorithm. Again, the data set is given on a uniform grid, this time of size 100\;$\times$\;100.

As the visualization of flows on surfaces has recently become a focus of research in the field of topology-based vector field visualization \cite{laramee09btbfva}, we chose as a fourth test case a tangential vector field on a torus given on a uniform grid of size 50\;$\times$\;50 in parameter space and computed its topology with our algorithm.

\begin{landscape}
	
\begin{table*}
\centering
 \caption{Occurrences of critical points in test data sets. Relative occurrences
 are rounded to three decimal places.}
 \label{tab:dataresults}
 \centering
 \begin{tabular}{@{}|l@{}l|r|r|r|r|r|r|r|r@{}l|@{}}
 \hline
 &Class&\multicolumn{2}{c}{Random field}&\multicolumn{2}{c}{Ocean flow}&\multicolumn{2}{c}{Air flow}&\multicolumn{2}{c}{Torus}&\\
 &&absolute &relative & absolute & relative &absolute & relative &absolute & relative &\\\hline\hline
 &\multicolumn{9}{c}{Overall}&\\
 \hline\hline
 &Topology-dependent & 62887918 & 0.629 & 11089 & 0.990 &9972&0.997&2492&0.997&\\\hline
 &Value-dependent    & 37112082 & 0.371 &  111  & 0.010 &28&0.003&8&0.003&\\\hline
 &Boundary point &0    & 0.000 &0      & 0.000 &0&0.000&0&0.000&\\\hline\hline
 &\multicolumn{9}{c}{Number of critical points} &\\\hline
 \hline
 &None                & 64094290 & 0.641 & 11141 & 0.995 &9984&0.998&2496&0.998&\\\hline
 &At least one        & 35905710 & 0.359 & 59    & 0.005 &16&0.002&4&0.002&\\\hline
 &Exactly one         & 33202155 & 0.332 & 57    & 0.005 &16&0.002&4&0.002&\\\hline
 &Exactly two         & 2087768  & 0.021 & 2     & 0.000 &0&0.000&0&0.000&\\\hline\hline
 &\multicolumn{9}{c}{Topology-dependent cases} &\\\hline
 \hline
 &No critical point    & 29685763  & 0.472 & 11032 & 0.985 &9956&0.996&2488&0.995&\\\hline
 &Saddle               &16600916   & 0.264 & 33    & 0.003 &4&0.000&2&0.001&\\\hline
 &Non-saddle           &16601239 & 0.264 & 24  & 0.002 &12&0.001&2&0.001&\\\hline
 &Saddle \& non-saddle & 0          & 0.000  & 0   & 0.000 &0&0.000&\\\hline\hline
 &\multicolumn{9}{c}{Value-dependent cases} &\\\hline\hline
 &No critical point    & 34408527   & 0.927 & 109  & 0.010 &28&0.003&8&0.003&\\\hline
 &Saddle \& non-saddle & 2087768    & 0.056 & 2    & 0.000 &0&0.000&0&0.000&\\\hline
 &Higher-order         & 615787     & 0.017 & 0    & 0.000 &0&0.000&0&0.000&\\\hline
\end{tabular}
\end{table*}

\end{landscape}

Table~\ref{tab:dataresults} documents the statistics with regard to
number and types of critical points for all test cases. One
interesting observation is that, although higher-order critical
points might occur, their frequency is very low (of course depending on the threshold value for $\Delta$). On the other hand,
two critical points within a cell are (at least in the case of the oceanic flow) more common  and need to be considered in practice. The most interesting finding is that the realistic ocean and air flow data sets contain only very few value-dependent cases (less than 1 percent). Therefore, our algorithm can in both cases identify
and classify more than 99 percent of the critical points by just
analyzing the topology of 0-value points on the cell
boundary---without the need to evaluate the Jacobian, solve a
quadratic system, or apply the sector method~\cite{tricoche00tsm}. Critical points on the boundary did not occur in this data series such that one can assume that these occur extremely seldom. Even if the quantization of the values for the random vector field test case is decreased to $10^{-4}$, only $46$ of the $10^8$ generated cases, i.e. $0.000046\%$, had critical points on cell boundaries. Generally speaking: The coarser the quantization or numeric accuracy of the computations the likelier the occurrence of critical points on cell boundaries becomes. None the less we believe their occurrence to be very seldom in practice.

Figure~\ref{fig:topologyocean} shows qualitative results for the
ocean data set, illustrating the difference between a linearized
version of the vector field and the original bilinear version:
Fig.~\ref{fig:topologyocean}(a) shows the (incorrect) topological
skeleton obtained from the linearized vector field and
Fig.~\ref{fig:topologyocean}(b) shows the correct version produced
by our method. The differences arise for va\-lue-de\-pen\-dent cells that
contain two critical points missed when linearizing the 0-isolines
of the vector field's components.

In terms of flows on surfaces an example in form of a vector field on the torus is examined. The standard $2$-torus $T^2$ can be parametrized via
\begin{equation*}
f(u,v)=\left((a+b\cos u)\cos v,(a+b\cos u)\sin v,b\sin u\right),
\end{equation*}
with $u,v,a,b\in \R$, $0\leq u,v<2\pi$, $0<b<a$.
Using this pa\-ra\-me\-tri\-za\-tion and its Jacobian $D{f}$ one can project two-di\-men\-sio\-nal vector fields defined in parameter space $[0,2\pi)^2\subset\R[2]$ and their topology obtained with our algorithm to the torus, see Figure~\ref{fig:torusvf}.

%-------------------------------------------------------------------
\section{Conclusions and Future Work}
 \label{sec:conclusion}

We have presented a novel approach to finding and classifying
critical points according to their Poincar\'{e} index for
barycentrically and bilinearly interpolated vector fields on triangular and rectilinear grids in parameter space, respectively. Our
approach is cell-based and efficient through the use of lookup
tables reminiscent of the classification of scalar fields by the
marching cubes algorithm. Our algorithm is able to deal with
critical points on cell boundaries, to detect second-order critical
points for bilinearly interpolated vector fields (together with
providing a measure for the stability of such critical points), and
to simplify vector field topology inside a cell by substituting two
first-order critical points with one of second order. We have
demonstrated that, for practical applications, the number of
critical points and their Poincar\'{e} indices can be identified by
just examining the intersection topology of the 0-level sets of the
interpolants with the cell boundaries. We have put our algorithm on
a sound mathematical foundation and showed how group theoretic tools
and a combinatoric description via cell colorings can be used to
solve the problem of critical-point classification, which may not
seem to be of a combinatoric nature at first glance.

As already pointed out in Section~\ref{sec:results}, the topology-based visualization of flow on surfaces has become a prominent field of research during the last years (c.f. \cite{laramee09btbfva}) and it thus would be of interest to be able to apply the algorithm presented in this paper to vector fields on triangulations and quadrangulations of arbitrary surfaces. A simple proof of concept was given in the case of the torus in Section~\ref{sec:results}. Note that in the quadrangular case the algorithm is well suited to be used in combination with the surface parametrization algorithm QuadCover by Kälberer et al. \cite{kalberer07qc} due to its structure and it thus seems a natural step to combine these two methods---this is research in progress.

More fundamental open questions for future work include: can this
method be extended to higher dimensional domains or to other
interpolants such as tensor-product cubic? More generically, it
might be of interest for other areas of visualization to see if they
might benefit from combinatoric and group-theoretic approaches.

\bigskip

Felix Effenberger\\
              Universität Stuttgart\\
					Institut für Geometrie und Topologie\\
              \url{effenberger@mathematik.uni-stuttgart.de}
              \\[1ex]
         Daniel Weiskopf \\
				Universität Stuttgart\\
				Visualization Research Center (VISUS)\\
           \url{weiskopf@vis.uni-stuttgart.de}

\bibliographystyle{alpha}
\bibliography{bibliography}   

\end{document}